\documentclass[11pt]{article}
\usepackage{amsmath}
\usepackage{amsthm}
\usepackage{amssymb}
\usepackage{algorithm}
\usepackage{subfig}
\usepackage{color}
\usepackage[english]{babel}
\usepackage{graphicx}
\usepackage{grffile}
\usepackage{url}
\usepackage{color}
\usepackage{algpseudocode}
\usepackage[T1]{fontenc}
\usepackage{dsfont}
\usepackage{thm-restate}
\usepackage{tcolorbox}

\def\isArxiv{1}

\usepackage[margin=1in]{geometry}

\usepackage{enumerate}

\newtheorem{theorem}{Theorem}[section]
 
\newtheorem{lemma}[theorem]{Lemma}
\newtheorem{definition}[theorem]{Definition}

\newtheorem{corollary}[theorem]{Corollary}

\newtheorem{fact}[theorem]{Fact}
\newtheorem{remark}[theorem]{Remark}
\newtheorem{claim}[theorem]{Claim}

\newcommand{\eps}{\epsilon}

\newcommand{\R}{\mathbb{R}}

\renewcommand{\varepsilon}{\epsilon}

\renewcommand{\bar}{\overline}
\renewcommand{\eps}{\epsilon}

\newcommand{\SAT}{\mathsf{SAT}}
\newcommand{\MAX}{\mathsf{MAX}}

\newcommand{\val}{\textsf{value}}

\newcommand{\xnote}[1]{{\color{red}[Xi: #1]}}
\newcommand{\Thomas}[1]{{\color{green}[Thomas: #1]}}

\makeatletter
\newcommand*{\RN}[1]{\expandafter\@slowromancap\romannumeral #1@}
\makeatother

\title{Computational Hardness of the Hylland-Zeckhauser Scheme}
\ifdefined\isArxiv
\author{Thomas Chen \\ Columbia Univeristy \\ \texttt{tc3039@columbia.edu}
	\and Xi Chen\\ Columbia University\\ \texttt{xichen@cs.columbia.edu} 
	\and Binghui Peng\\ Columbia University\\ \texttt{bp2601@columbia.edu}
	\and Mihalis Yannakakis  \\ Columbia University \\ \texttt{mihalis@cs.columbia.edu}
}
\fi

\begin{document}
\maketitle

\begin{abstract}
We study the complexity of the classic Hylland-Zeckhauser scheme \cite{hz79} for one-sided matching markets.
We show that the problem
of finding an $\eps$-approximate equilibrium in the HZ scheme is PPAD-hard, and this holds even when $\eps$ is polynomially small
  and when each agent has no more than four distinct utility values.
Our hardness result, when combined with the PPAD membership result of \cite{vy21}, resolves the approximation complexity of the HZ scheme.
We also show that the problem of approximating the optimal social welfare (the weight of the matching) 
  achievable by HZ equilibria within a certain constant factor is NP-hard.

\end{abstract}

\newpage

\section{Introduction}
\label{sec:intro}

In a {\em one-sided matching} problem, there is a set $A$ of $n$ agents and a set $G$ of $n$ goods, and we are given a specification of the preferences of the agents for the goods.\footnote{In the statement of the problem, we have the same number of agents and goods for simplicity. In general there can be $n_1$ agents, each with their own demand $d_i$, and $n_2$ goods, each with its own supply, $s_j$,
where $\sum_{i \in n_1} d_i = \sum_{j \in n_2} s_j$. 
It is clear that this setting reduces to the simpler case of equal numbers of agents and goods.}
The problem is to find a matching between the agents and the goods (assigning a distinct good to each agent) that has desirable properties, such as Pareto optimality, envy-freeness, incentive compatibility.
This situation, where only one side has preferences, arises in many settings, such as assigning students to schools, assigning faculty members to committees, workers to tasks, program committee members to papers, students to courses with limited capacity, etc.

Since many agents may have the same or similar preferences, it is usually not possible to offer everybody their favorite good. So a solution mechanism has to strive to be equitable, satisfy the agents as much as possible, and incentivize them to give their true preferences (i.e., not gain an advantage by lying).
Randomization is often useful to meet fairness requirements.
A randomized solution mechanism has probability $x_{i,j} \in [0,1]$ of matching each agent $i$ to each good $j$; these probabilities form a doubly stochastic matrix, i.e., a fractional perfect matching in the bipartite graph between agents and goods.
In some applications, the goods may be divisible, or they may represent tasks or resources that can be shared among agents;
in these cases the quantities $x_{i,j}$ represent the shares of the agents in the goods. 

There are two main ways of specifying the preferences of each agent $i \in [n]$ for the goods: (1) {\em cardinal preferences}, where we 
are given the utility $u_{i,j}$ of agent $i$ for each good $j \in [n]$,
or (2) {\em ordinal preferences}, where we are given the agent's total
ordering of the goods.
Cardinal preferences allow for a finer specification of the agents' preferences (although they may require more effort to produce them).
As a result, they can yield better assignments.
Consider for instance the following example from \cite{hz79}: There are 3 agents and 3 goods. The utilities of agents 1 and 2 for the three goods are 100, 10, 0, while agent 3 has utilities 100, 80, 0.
The ordinal preferences of the three agents are the same, so any fair mechanism will not distinguish between them, and will give them each probability 1/3 for each good.
The expected utilities of the three agents in this solution is
$36 \frac{2}{3}, 36 \frac{2}{3}, 60$.
This solution is not Pareto optimal, i.e., there is another solution where all agents are better off. Agent 3 is assigned good 2, and agents 1, 2 randomly split goods 1 and 3. The expected utilities of the three agents in this solution are $50, 50, 80$.

In 1979 Hylland and Zeckhauser proposed a, by now, classic scheme for the one-sided matching problem under cardinal preferences \cite{hz79}.
The scheme uses a pricing mechanism to produce an assignment of probability shares $\{ x_{i,j} | i,j \in [n] \}$ of goods to agents, 
i.e. a fractional perfect matching, and these are then used in
a standard way to generate probabilistically an integral perfect matching. The basic idea is to imagine a market where
every agent has 1 dollar, and the goal is to find prices for the goods and (fractional) allocations $x_{i,j}$ of goods to the agents, 
such that
the market clears (all goods are sold), while every agent maximizes
her utility subject to receiving a bundle of goods of size 1 and cost at most 1. Hylland and Zeckhauser showed that such an equilibrium
set of prices and allocations always exists, using Kakutani's fixed point theorem. Note that money here is fictitious; no money changes hands. The only goal is to produce the allocation (the shares $ x_{i,j}$) so that it reflects the preferences of the agents. The HZ scheme has several desirable properties: it is Pareto optimal, envy-free \cite{hz79}, and it is incentive compatible in the large \cite{he18}.  The scheme has been extended and generalized in various ways since then.\newpage

Although the HZ scheme has several nice properties, one impediment is that, despite much effort, there is no efficient algorithm known to compute an equilibrium solution. This has remained an open problem till now. In \cite{akt17}, Alaei, Khalilabadi and Tardos gave polynomial-time algorithms for the case that the number of goods or the number of agents is a fixed constant (the case of a constant number of goods can be derived also from \cite{dk08}). Recently in \cite{vy21}, Vazirani and Yannakakis gave a polynomial-time algorithm for the bi-valued case, where every agent's utilities take only two values.
They also gave an example showing that the equilibrium prices and allocations can be inherently irrational. In the general case, they showed that the problem of computing an equilibrium solution is in the class FIXP. Furthermore, computing an $\epsilon$-approximate equilibrium is in the class PPAD, where in an approximate equilibrium an agent may get a slightly suboptimal allocation and may spend $1+\epsilon$ dollars. They leave open the problem whether computing an exact or approximate equilibrium is complete for the classes.

In this paper we resolve the complexity of computing an approximate equilibrium
  of the HZ scheme. Our main result is:

\begin{theorem}[Main]
	\label{thm:hz-equilibrium}
	The problem of computing an $\eps$-approximate equilibrium
	  of the HZ scheme is PPAD-complete 
	  when $\eps=1/n^c$ for any constant $c>0$.
\end{theorem}

In our construction, every agent has at most 4 different utilities for the goods. Thus, the problem is PPAD-complete even for 4-valued utilities\footnote{The case of a small number of values is natural. For example, the authors have been in committees that ask them to
rate their level of interest in the submissions by values in a limited range, e.g. 0-4.}. We leave the 3-valued case open. We give however a simple example with values in $\{0, 0.5, 1 \}$ showing that there can be multiple disconnected equilibria, thus suggesting that usual convex programming methods may not work (at least a convex program will not include all equilibria).

A given instance of the one-sided matching problem may have multiple HZ equilibria.
All of them are Pareto optimal, but some may be preferable to others when  other criteria are considered.
One such criterion is the social welfare, i.e., the total weight of the matching
  (or the sum of utilities of agents).
We study the problem of approximating the optimal social welfare achievable
  by an HZ equilibrium. 
We show that this is an NP-hard problem:

\begin{theorem}
	\label{thm:social-welfare}
   Given an instance of the one-side matching problem and a value $w$, it is NP-hard to distinguish the case that the maximum social welfare of an HZ equilibrium is at least $w$ from the case that it is at most
$(\frac{175}{176} + \eps)w$ for any constant $\eps>0$.
\end{theorem}


\subsection{Proof Overview}
\label{sec:overview}

We give an overview of the proof of Theorem \ref{thm:hz-equilibrium}. 
To prove the problem of finding an approximate HZ equilibrium is PPAD-hard, we 
  give a polynomial-time reduction from 
   {\em threshold games}, introduced recently by Papadimitriou and Peng \cite{pp21}.
A threshold game is defined on a directed graph $G = (V, E)$, with a variable $x_v \in [0,1]$ associated with each node $v \in V$. The equilibrium condition is characterized by a comparison operator: $x_v = 1$ if $\sum_{ (u, v) \in E} x_u \leq 1/2 -\kappa$; $x_v = 0$ if $\sum_{ (u, v) \in E} x_u \geq 1/2 +\kappa$ and $x_v$ can take an arbitrary value in $[0,1]$, otherwise.
The PPAD-hardness of threshold game is proved in \cite{pp21}, and it holds for some positive constant $\kappa > 0$ and for sparse graphs.
As we will see later, the use of threshold games significantly simplifies the reduction.

From a high-level view, our reduction follows the general framework 
  of previous hardness results on market equilibria \cite{cddt09,ct09,vy11,cpy17}:
we use  prices of an HZ market to simulate variables $x_v$ in a threshold game
  and the construction is based on the design of two gadgets:
  \emph{variable gadgets} for each $v\in V$ (to simulate variables $x_v$ and enforce the equilibrium condition at each node $v$) and \emph{edge gadgets} for each $e=(u,v)\in E$ (to simulate the action of sending 
  $x_u$ to the sum at $x_v$ in the threshold game). 
However, 
a major challenge of working with the HZ scheme is that it is difficult to characterize the equilibrium behavior of agents in this model, as it is complex, nonlinear, and does not admit a closed form solution. As a consequence, it is hard both, to analyze even small instances, and to synthesize instances with desired characteristics. 
Below we discuss some of the key ideas behind the construction.



\vspace{+2mm}
{\bf \noindent Variable gadgets.}
To simulate variables $x_v$ of a threshold game,
  our starting point is the following simple sub-market.
There are two agents and two goods. Both agents have utility $1$ for good $1$ and utility $\frac{1}{2\delta^{-1} - 1}$ for good $2$ ($0 < \delta < 1$ which should be considered as a small constant as discussed below). 
It is easy to observe that the set of equilbrium prices are $(p, 2-p)$ for $p \in [0,\delta]$.
After scaling  by $1/\delta$,
  the price $p$ of good $1$ in this sub-market can be used to simulate 
  a variable $x_v\in [0,1]$ in the threshold game.
So we create such a sub-market $M_v$ for each node $v$ and 
  denote the price of good $1$ in $M_v$ by $p_v$.
To finish the reduction, it suffices to create agents that 
  are interested in goods in $M_u$ and $M_v$, 
  for each edge $e=(u,v)\in E$, such that
  the total allocation of goods from $M_v$ to them is captured by 
  $-p_u$.\footnote{The negative allocation
  may look strange but can be achieved (essentially) by offsetting the supply carefully.} 
The agents created for this task are what we referred to earlier as the edge gadget for $e$.
If achieved, the total allocation of goods of $M_v$ to agents outside
  would give the desired linear form $-\sum_{(u,v)\in E} p_u$ which is 
  a scaled version of $-\sum_{(u,v)\in E} x_u$.
The sub-market $M_v$ can help enforce the equilibrium condition of the
  threshold game.
When $\sum_{(u,v)\in E} x_u$ is too small (so the allocation to agents outside 
  is high due to the negative sign),  there is a 
  shortage of goods of $M_v$, which 
  would lead to $p_v\approx \delta$ and thus, $x_v\approx 1$; on the other hand, if the sum is too large, then there is a surplus of goods in the sub-market $M_v$, which forces 
    $p_v\approx 0$ and thus, $x_v\approx 0$.

\vspace{+2mm}
{\bf \noindent Edge gadgets.}
The key technical challenge lies in the construction of edge gadgets.
Our first attempt is to create an agent who has utility $1/2$ for good 
  $1$ in $M_u$ and utility $1$ for good $2$ in $M_v$ (with price $2-p_v$).
The optimal bundle for this agent, however, is not easy to work with
  at first sight: for example, the agent is allocated
  $\frac{1 -p_u}{2 - p_u - p_v}$ unit of good $2$ in $M_v$. 
{{\bf The first key idea of our reduction is to use first order approximation to simplify a complex function}.
That is to say}, we set $\delta$ to be a sufficiently small constant and apply first-order
  approximation on the allocation.
Ignoring constant factors and constant or lower-order terms, 
  the agent described earlier has $D_{u}\approx  p_u - p_v$
    unit of good $1$ in $M_u$ and $D_{v} \approx p_v - p_u$ unit of good $2$
    in $M_v$ in her optimal bundle.
This, however, is far from what we hoped for,  as (1) we don't want $p_v$ to appear in $D_v$, and (2) we want $D_u = 0$.
(Notably if we use this agent as our edge gadget for every $e\in E$, then they together
  are essentially simulating a threshold game over an undirected graph, which 
  admits trivial equilibria.)
What about agents with a different set of utilities for goods in $M_u$ and $M_v$?
Perhaps surprisingly, all our attempts fail and there is a fundamental reason for that: it can be shown that, no matter how the utilities are set, the allocation of goods in $M_v$ is always monotonically increasing with $p_v$, which is undesirable for our purpose. 

We circumvent this obstacle using the following three steps.
First, we introduce extra goods with a fixed price $2$ into the picture
  (where the fixed price can be enforced easily by creating agents who
  are interested in these goods only).
Second, we replace the current variable gadget with a richer 
  sub-market with three goods, with the price of the new good 
  set to be $\frac{(1+p_v)}{2}$.
Agents in our edge gadget can now have access to these new goods
  which give us a larger design space for their utilities and equilibrium behavior.
The last step, which is {\bf the second key idea of our reduction, is to use discrete functions to approximate a continuous function}.
 This is essential in our proof and we believe it might be of independent interest in reductions of similar settings.
 Concretely, we design agents that perform 
  comparison operations: there are two possible optimal bundles for each
  of these agents, and which one it is depends on the sign (in a robust sense)
  of a certain affine linear form
  of $p_u$ and $p_v$.
  The agent behaves like a step function, which is not useful on its own. However, when combined, one can construct a series of agents by enumerating utilities; these lead to careful cancellations that make sure
  the total allocation $D_u$ of goods from $M_u$ is $0$ as desired.

\subsection{Related Work}
\label{sec:relate}

We have already mentioned the most relevant work on the complexity of the Hylland-Zeckhauser scheme. The problem of computing an exact HZ equilibrium is in FIXP, and computing an approximate equilibrium is in PPAD \cite{vy21}. Polynomial-time algorithms for a fixed number of agents~or goods were given in \cite{akt17}. It has been a longstanding open problem about whether there is a polynomial-time algorithm in the general case.

The input in the one-sided matching problem is the same as in the
classical assignment problem (equivalently, maximum weight perfect matching problem in bipartite graphs). This is one of the most well-studied problems in Operations Research and Computer Science, and several very efficient algorithms have been developed for it over the years. The difference in the one-sided matching problem is that the primary consideration is to produce a solution that has certain desirable fairness and optimality properties for the agents;  the goal is not simply the maximization of the weight of the matching.
As we show in this paper, computing an HZ solution for the one-sided matching problem is  probably computationally harder: it is PPAD-hard
to compute any approximate HZ solution, and if we want to maximize the total weight of the matching as a secondary criterion then the problem becomes NP-hard.

In the case of ordinal preferences for the agents in the one-sided matching problem, there are other schemes with nice properties: the Random Priority (also called Random Serial Dictatorship) scheme \cite{as98, moulin18} and the Probabilistic Serial scheme \cite{bm01}.
These have polynomial-time (randomized) algorithms. However, since they are based only on ordinal preferences, they are suboptimal with respect to the agents' utilities, as the earlier simple example shows. 

The setting in the HZ scheme is the same as in the linear Fisher market model: the input consists of the utilities $u_{i,j}$ of the agents for the goods, and the problem is to compute equilibrium prices and allocations. The only difference is that when an agent picks her optimal bundle of goods, she must get exactly one unit (in addition to the cost being within the budget of 1 dollar), i.e. the solution must be a (fractional) perfect matching. Although this may seem like a small difference, it has a substantial effect both in the structure of the problem and in its computational complexity: exact solutions may be irrational, and as we show in this paper, finding an approximate solution is PPAD-hard. The linear Fisher model has been studied extensively and there are polynomial-time algorithms for computing equilibria in this model, as well as in the more general Arrow-Debreu model with linear utilities \cite{dpsv08,orlin10,jain07}.

There is furthermore extensive work on markets with more complex utility functions than linear, such as piecewise linear, Leontief, CES utilities and others, and for many of them it is PPAD-hard or FIXP-hard to compute an approximate or exact equilibrium
(e.g. \cite{cddt09,ct09,vy11,cpy17,ey10,gmv17}). 

Several researchers have proposed Hylland-Zeckhauser-type mechanisms for a number of appli\-cations, e.g. see \cite{budish11,he18,le17,mclennan18}. There are also recent works that have generalized and extended the basic HZ scheme in several directions, for example  to two-sided matching markets and to an Arrow-Debreu-type setting where the agents own initial endowments \cite{emz19a,emz19b,gtv20}. Note that in the case of initial endowments, an HZ equilibrium may not always exist, so some approximation or slack is needed to ensure existence
(see \cite{emz19a,gtv20}). 




\section{Preliminaries}
\label{sec:pre}

We write $[n]$ to denote $\{1,2,\ldots, n\}$. Given two integers $n$ and $m$ we use $[n:m]$
  to denote integers between $n$ and $m$, with $[n:m]=\emptyset$ when $m<n$.
Given two real number $x, y \in \R$, we use $x = y\pm \eps$ to denote $x\in [y-\eps, y+\eps]$.

\subsection{The Hylland-Zeckhauser Scheme}
We provide a formal description of the Hylland-Zeckhauser scheme for one-sided matching markets
  \cite{hz79, vy21}. 
It will be convenient for us to describe it using the language of linear Fisher markets.
An \emph{HZ market} $M$ consists of  
  a set $A=[n]$ of $n$ agents and a set $G = [n]$ of $n$ (infinitely) divisible goods. 
Each agent $i\in A$ has one dollar and there is one unit of each good $j\in G$ in the market.
We write $u_{i,j}\in [0,1]$\footnote{As it will become clear in Definition \ref{def:hz-equilibrium},
shifting and scaling utilities of agents does not change the set of HZ equilibria.
We assume utilities to lie in $[0,1]$ because we will consider an additive approximation
  of HZ equilibria in Definition \ref{def:hz-approx-equilibrium}.}
 to denote the utility of one unit of good $j$ to agent $i$, for each $i\in A$ and $j\in G$.
Hence an HZ market $M$ is  specified by a positive integer $n$ and utilities $(u_{i,j}:i,j\in [n])$.

Given an HZ market $M$ with $n$ agents and goods, an \emph{HZ equilibrium} \cite{hz79}
  consists of  an allocation $x=(x_{i,j}:i,j\in [n])$ and a price vector $p=(p_j:j\in [n])$
  that are nonnegative and satisfy a list of properties to be described in Definition \ref{def:hz-equilibrium}.
Given $x$ and $p$, we will refer to $x_i=(x_{i,j}:j\in [n])$ as the \emph{bundle of goods allocated to agent $i$}.
The \emph{cost} of the bundle $x_i$ is given by $\sum_{j\in [n]} p_j x_{i,j}$ and the \emph{value}  of $x_i$ to agent $i$ is 
  $\sum_{j\in [n]} u_{i,j} x_{i,j}$.
We are ready to define HZ equilibria:


\begin{definition}[HZ Equilibria \cite{hz79}]
	\label{def:hz-equilibrium}
A pair $(x,p)$, where $x=(x_{i,j}:i,j\in [n])\in \mathbb{R}_{\ge 0}^{n\times n}$ and $p=(p_i:i\in [n])\in \mathbb{R}^n_{\ge 0}$, 
  is an \emph{HZ equilibrium} of an HZ market $M$ if: 
	\begin{flushleft}\begin{enumerate}
		\item The total allocation of each good $j\in [n]$ is $1$ unit, i.e., $\sum_{i\in [n]}x_{i,j} =1$.
		\item The total allocation of each agent $i\in [n]$ is $1$ unit, i.e., $\sum_{j\in [n]}x_{i,j} =1$.
		\item The cost of the bundle $x_i$ of each agent $i\in [n]$ is at most $1$, i.e.,
		$\smash{\sum_{j \in [n]}p_jx_{i,j} \leq 1}$.
		\item For each $i\in [n]$, $x_i$ maximizes its value 
		$\sum_{j\in [n]} u_{i,j}x_{i,j}$ to agent $i$ subject to 2 and 3 above.\footnote{We note that in \cite{hz79, vy21},
		  $x_i$ is required (as a tie-breaking rule) to minimize its cost among all those that maximize the value subject to items 2 and 3.
This is needed to ensure Pareto optimality of the equilibrium allocations. However, we 		
do not need this condition for our hardness results, and this only makes the results stronger. So for simplicity, we omit the condition
from the definition of exact and approximate equilibria.}
	\end{enumerate}\end{flushleft}
\end{definition}

Equivalently the last condition in the definition above can be captured by the following LP:
\begin{align*}
 \text{\emph{maximize}}\ \ &\sum_{j\in [n]}  u_{i,j}x_{i,j}\\ \text{\emph{s.t.}}\ \ 
&\sum_{j\in [n]}x_{i,j} = 1,\ 
\sum_{j\in [n]}p_j x_{i,j} \leq 1,\ \text{and}\ 
 x_{i,j} \geq 0\ \text{for all $j\in [n]$}.
\end{align*}
Taking $\mu_i$ and $\alpha_i$ to be the dual variables, one has the following dual LP that will be useful:
\begin{align*}
\text{\emph{minimize}}\ \ \hspace{0.05cm}&\alpha_i + \mu_i\\[0.3ex]
\text{\emph{s.t.}}\ \ \hspace{0.05cm}&\alpha_i \geq 0\ \text{and}\ 
\alpha_i p_j + \mu_i \geq u_{i,j},\ \text{for all $j\in [n]$}.
\end{align*}
We will refer to the LP (and its dual LP) above as the LP (or dual LP) for 
  agent $i$ with respect to the price vector $p$. 
Let $\val_p(i)$ denote their optimal value. Then it captures the optimal value of any bundle of goods
  to agent $i$ subject to conditions 2 and 3 in Definition \ref{def:hz-equilibrium}. 

Hylland and Zeckhauser \cite{hz79} showed that an HZ equilibrium always exists:

\begin{theorem}[Existence \cite{hz79}]
Every HZ market admits an HZ equilibrium.
\end{theorem}


If $(x,p)$ is an equilibrium, then it is easy to see that if we scale
the difference of all prices from 1, the resulting price vector $p'$
together with the same allocation $x$ forms also an equilibrium;
i.e. for any $r>0$ with $r \leq \min \{ 1/(1-p_j) | p_j <1 \}$, setting $p'_j = 1 + r(p_j -1)$ for all $j \in [n]$ yields a vector
$p'$ such that $(x,p')$ is also an equilibrium (see \cite{vy21}).
The reason is that this scaling does not affect the set of feasible allocations, as can be easily seen, and
$\val_{p'}(i)=\val_p(i)$ for all agents $i \in [n]$.
That is, price vectors related to each other by this scaling are in a sense equivalent. A consequence of this observation
is that we may always assume w.l.o.g. that an equilibrium contains a good with price 0 \cite{hz79}:
If one of the goods has price $<1$, then we can always scale the prices so that the minimum price is 0.
On the other hand if all prices in an equilibrium are $\ge 1$,
then all prices must be 1 (the sum of the prices must be $\leq n$,
the sum of the agents' budgets), and in this case the cost
condition 3 is redundant, and the all-0 vector forms also an equilibrium with the same allocation.
 
We say that a price vector $p$ is {\em normalized} if $\min_i p_i =0$.
We will restrict our attention henceforth to normalized price vectors, without always mentioning it explicitly.

Our hardness results hold for the following relaxation studied by 
  Vazirani and Yannakakis \cite{vy21}:

\begin{definition}[Approximate HZ Equilibria]
	\label{def:hz-approx-equilibrium}
Given some $\eps>0$, a pair $(x,p)$, where $x=(x_{i,j}:i,j\in [n])\in \mathbb{R}_{\ge 0}^{n\times n}$~and $p=(p_i:i\in [n])\in \mathbb{R}^n_{\ge 0}$ (where $\min_{i\in [n]} p_i =0$)\footnote{The requirement that $p$ be normalized is important in the definition because otherwise condition 3 on the cost has no effect: if $(x,p)$ is any pair that satisfies conditions 1,2,4, then we can always scale $p$ as above 
to a vector $p'$
where all prices are sufficiently close to 1 so that condition 3 is also satisfied for $(x,p')$.},
  is an \emph{$\eps$-approximate HZ equilibrium} of an HZ market $M$ if: 
	\begin{flushleft}\begin{enumerate}
			\item The total allocation of each good $j\in [n]$ is $1$ unit, i.e., $\sum_{i\in [n]}x_{i,j} =1$.
		\item The total allocation of each agent $i\in [n]$ is $1$ unit, i.e., $\sum_{j\in [n]}x_{i,j} =1$.
		
			\item  The cost of $x_i$ is at most $1+\eps$ for each   $i\in [n]$, i.e.,
			  $\smash{\sum_{j \in [n]}p_jx_{i,j} \leq 1+\eps}$.
			\item The value $\sum_{j\in [n]} u_{i,j} x_{i,j}$ of $x_i$ to agent $i$ is at least $\emph{\val}_p(i) - \eps$ for 
			each $i\in [n]$. 
	\end{enumerate}\end{flushleft}
\end{definition}

An alternative, more relaxed notion of an $\eps$-approximate equilibrium, where condition 1 is also relaxed to $|\sum_{i\in [n]}x_{i,j} -1| \leq \eps$ for all goods $j \in [n]$, is polynomially equivalent to the above notion \cite{vy21}. Thus, it follows that computing an $\eps$-approximate equilibrium under the more relaxed
notion is also PPAD-complete.


\subsection{Threshold Games}
Our PPAD hardness results use \emph{threshold games}, introduced recently by Papadimitriou and Peng \cite{pp21}.
They showed that the problem of finding an approximate equilibrium in a threshold game is PPAD-complete. %
\begin{definition}[Threshold game \cite{pp21}]
	\label{def:threshold}
	A threshold game is defined over a directed graph $H = (V, E)$. Each node $v\in V$ represents a player with strategy space $x_v\in [0,1]$. 
	Let $N_v$ be the set of nodes~$u\in V$ with $(u,v)\in E$.
Then $x = (x_v:v\in V) \in [0, 1]^{V}$ is a $\kappa$-\emph{approximate equilibrium} if every $x_v$ satisfies
%
	\begin{align*}
	x_v \in  
	\begin{cases}
	[0,\kappa] & \sum_{u \in N_v} x_u > 0.5 + \kappa\\
	[1-\kappa,1] & \sum_{u \in N_v} x_u < 0.5 - \kappa\\
	[0,1] & \sum_{u \in N_v}x_u \in[ 0.5 - \kappa, 0.5+\kappa] 
	\end{cases} 
	\end{align*}
\end{definition}

\begin{theorem}[Theorem 4.7 of \cite{pp21}]
	\label{thm:ppad-threshold}
	There is a positive constant $\kappa$ such that the problem of 
  finding a $\kappa$-approximate equilibrium in a threshold game is PPAD-hard. 
This holds even when the in-degree and out-degree of each node is at most $3$ in the threshold game.
\end{theorem}


\ifdefined\isArxiv
\else
\graphicspath{{./Figures/}}
\fi

\section{PPAD-hardness}
\label{sec:ppad}

Our goal in this section is to prove the following theorem:

\begin{theorem}\label{theo:ppadhard}
The problem of finding a $(1/n^5)$-approximate HZ equilibrium in an HZ market 
  with $n$ agents and goods is PPAD-hard. 
\end{theorem}

 In Appendix \ref{sec:padding} (via a standard padding argument),
  we give a polynomial-time reduction from
  the problem of finding a $(1/n^5)$-approximate HZ equilibrium to
  that of finding a $(1/n^c)$-approximate HZ equilibrium in an HZ market, for any positive constant $c$.
Theorem \ref{thm:hz-equilibrium} follows by combining the PPAD membership result of \cite{vy21}.

Our plan is as follows. Let $\eps=1/n^5$ throughout this section wherever an HZ market with
  $n$~agents and goods is concerned.
We start with some basic facts about approximate HZ equilibria in~Section \ref{sec:basic}
  (mainly about how to work with approximately optimal
  bundles for agents).
Then we describe the polynomial-time reduction from threshold games to HZ markets 
  in Section \ref{sec:construction}. Our reduction constructs two types of gadgets,
  \emph{variable gadgets} and \emph{edge gadgets},
  which simulate variables $x_v$ and edges $(u,v)$ in a threshold game, respectively.
Using these gadgets, we finish the reduction's correctness proof in Section \ref{sec:correctness}; 
  the analysis of these two gadgets is presented afterwards in Section \ref{sec:gadgets} and \ref{sec:edgegadgets}, respectively.

\subsection{Basic Facts}
\label{sec:basic}
Let $M$ be an HZ market with $n$ agents and goods.
As it will become clear later, the HZ market we construct in the reduction satisfies
   $\max_{j\in [n]} u_{i,j}=1$ for every agent $i\in [n]$.
\emph{Hence we assume this is the case in every HZ market discussed in the rest of this section.}
In all lemmas of this subsection we assume $(x,p)$ to be an $\eps$-approximate HZ equilibrium
  of $M$ (and skip it in their statements).
Recall that prices are normalized: $\min_i p_i =0$.

We give first an upper bound on the sum of prices: 

\begin{lemma}
	\label{fact:maximum-price}
	$\sum_{j\in [n]} p_j\le 2n$.
\end{lemma}
\begin{proof}
Since $(x,p)$ is an $\eps$-approximate HZ equilibrium, every good must be sold out 
  and no agent can spend more than $1+\eps$. Thus, $\sum_{j\in [n]} p_j\le n(1+\eps)< 2n$
  using $\eps=1/n^5$. 
\end{proof}

Next, we consider an optimal solution $(\alpha_i^*,\mu_i^*)$ to the dual LP for agent $i$
  and prove the following: 

\begin{lemma}\label{simplelemma1}
$\mu_i^*\ge 0$ and $\alpha_i^*\le 1$ for every $i\in [n]$.
\end{lemma}
\begin{proof}
Let $\ell$ be a good with $p_\ell=0$. From the dual LP constraints, we have 
  $ 0 \leq u_{i,\ell} \leq \alpha_i^* p_\ell + \mu^*_i = \mu_i^*.$ 
Moreover, since all utilities are in $[0,1]$ we have trivially that
$\alpha_i^*+\mu_i^*=\val_p(i)\le 1$. 
Therefore, we have $\alpha_i^* \leq 1$.
\end{proof}
\begin{lemma}\label{heheclaim1}
If $\emph{\val}_p(i)\le 0.9$ then $\alpha_i^*\ge 1/(20n)$
  and $\sum_{j\in [n]}p_j x_{i,j}\ge 1-20n\eps$.
\end{lemma}
\begin{proof}
Let $u_{i,\ell}=1$. It follows from Lemma \ref{fact:maximum-price} that $p_\ell\le 2n$ and thus,
  $$1=u_{i,\ell}\le \alpha_i^* p_\ell+\mu_i^*\le 2n\alpha_i^*+\mu_i^*.$$
On the other hand, $\val_p(i)=\alpha^*_i+\mu^*_i\le 0.9$. The first part of the lemma follows from adding these two inequalities.

Next, multiplying both sides of the inequalities 
$\alpha_i^* p_j + \mu_i^* \ge u_{i,j}$ by $x_{i,j}$,
summing over all $j \in [n]$, and using $\sum_j x_{i,j}=1$, we have
$$
\sum_{j\in [n]} \alpha_i^* p_jx_{i,j}+\mu_i^*\ge  \sum_{j\in [n]} u_{i,j}x_{i,j}
\ge \alpha_i^*+\mu_i^*-\eps.
$$
The second part of the lemma then follows from $\alpha_i^*\ge 1/(20n)$. 
\end{proof}

Recall that all goods $j$ satisfy $u_{i,j} \leq \alpha_i^* p_j + \mu_i^*$.
We say a good $j$ is $\delta$-\emph{suboptimal} for agent $i$ if
$u_{i,j} + \delta \leq \alpha_i^* p_j + \mu_i^*$.
We show that agent $i$'s good-bundle, $x_i$, cannot contain significant quantities of suboptimal goods.

\begin{lemma}
	\label{fact:subopt-goods}
For every $i\in [n]$, the 
total allocation in $x_i$ to $\delta$-suboptimal goods is at most $2 \epsilon / \delta$.
\end{lemma}
\begin{proof}
Fix an agent $i\in [n]$. 
We have $u_{i,j} \leq \alpha_i^* p_j + \mu_i^*$ for all $j\in [n]$, and 
  $u_{i,j}+\delta \le \alpha_i^* p_j+ \mu_i^*$ for $\delta$-suboptimal goods.
Let $W$ be the total allocation in $x_i$ to $\delta$-suboptimal goods.
Then $$\sum_{j\in [n]} u_{i,j} x_{i,j} + W \delta \leq \alpha_i^* \sum_{j\in [n]} p_j x_{i,j} + \mu_i^* \sum_{j\in [n]} x_{i,j}.$$
Using the definition of $\epsilon$-approximate HZ equilibria, 
  the LHS is at least 
$$\val_p(i)-\eps+W\delta=\alpha_i^*+\mu_i^*+W\delta-\eps$$
and the RHS is at most $\alpha_i^*(1+\eps)+\mu_i^*$.
The lemma follows from $\alpha_i^*\le 1$ by Lemma \ref{simplelemma1}.
\end{proof}

We use some of the lemmas above to obtain following corollaries: 

\begin{corollary}\label{maincoro}
Let $J$ be the set of $j\in [n]$ that are \emph{not} $\delta$-suboptimal for $i$. 
If $\emph{\val}_p(i)\le 0.9$, then 
\begin{align*}
1-2\eps/\delta\le \sum_{j\in J} x_{i,j}\le 1\quad\text{and}\quad
1-20n\eps -\frac{4n\eps}{\delta}\le 
  \sum_{j\in J} p_jx_{i,j}\le 1+\eps.
\end{align*}
\end{corollary}
\begin{proof}
The first part follows directly from Lemma \ref{fact:subopt-goods}.

The second part follows from Lemma \ref{fact:maximum-price}, Lemma \ref{heheclaim1}, 
  and Lemma \ref{fact:subopt-goods}.
\end{proof}

\begin{corollary}\label{secondcoro}
Let $J$ be the set of goods $j\in [n]$ with $u_{i,j}>0$.
If $\emph{\val}_p(i)\le 0.9$, then
$$
1-20n\eps-\frac{1}{n^2}\le \sum_{j\in J} p_jx_{i,j}\le 1+\eps.
$$
\end{corollary}
\begin{proof}
The second inequality clearly holds since $(x,p)$ is an $\epsilon$-approximate equilibrium. Suppose that the first inequality does not hold.
Then by Lemma \ref{heheclaim1}, agent $i$ spends more than
$1/n^2$ on zero-utility goods, hence she buys at least an amount $1/2n^3$ of these, since all prices are at most $2n$.
Consider a new bundle for $i$ obtained by replacing $1/2n^3$ of
the zero-utility goods by a good with utility 1.
The cost of the new bundle is still less than 1,
i.e. it is a feasible bundle, and
the value exceeds that of the original bundle $x_i$ by
$1/2n^3 > \epsilon$, contradicting the fact that
$(x,p)$ is an $\epsilon$-approximate equilibrium.
\end{proof}

Finally we include a simple lemma about the optimal value of an agent:
\begin{lemma}\label{simplesimplelemma}
Let $i\in [n]$ and $\ell\in [n]$ with $u_{i,\ell}=1$. Then
  $\emph{\val}_p(i)\ge \min(1,1/p_\ell)$.
\end{lemma}
\begin{proof}
If $p_\ell=0$, then agent $i$ can get value $1$ by buying one unit of good $\ell$ for free.

If $p_\ell>0$ then there is another good with zero price and thus, agent $i$ can 
  get value $\min(1,1/p_\ell)$ by buying $\min(1,1/p_\ell)$ unit of good $\ell$ and 
  $1-\min(1,1/p_\ell)$ unit of a zero price good.
\end{proof}

\def\indeg{\text{in-deg}}
\def\outdeg{\text{out-deg}}

\subsection{The Construction}\label{sec:construction}

Let $\kappa\in (0,1)$ be the positive constant in Theorem \ref{thm:ppad-threshold}.
Recall that our goal is to give a polynomial-time reduction from the problem of finding a $\kappa$-approximate equilibrium
  in a threshold game (with both in-degree and out-degree at most $3$) to that of finding an $\eps$-approximate HZ
  equilibrium in an HZ market with $\eps=1/n^5$.

Let $C$ be a sufficiently large universal constant, and $m=\lceil C/\kappa \rceil$.
Let $H=(V, E)$ be a threshold game with $|V|=N$.
(Note that $N$ is asymptotically large and should be considered as larger than any function of $m$.)
We write $\indeg(v)$ and $\outdeg(v)\le 3$ to denote the in-degree and out-degree of $v\in V$, respectively.
We construct an HZ market $M_H$ from $H$ in three steps as described below.
This is done by creating \emph{groups} of goods and \emph{groups} of agents, with 
  the guarantee that agents in the same group have the same utility for any good as each other, and that goods in the same group yield the same utility to any agent.
We say a group~$A_i$~of~agents have utility $u$ for a group $G_j$ of goods if all agents
  in $A_i$ share the same utility $u$ for all goods in $G_j$.
(Intuitively we create a group $A_i$ of agents to  simulate an agent with demand and budget $|A_i|$ instead of $1$, and
   a group $G_j$ of goods to simulate a good with a supply of $|G_j|$ units in the market.
A technical subtlety though is that in an approximate HZ equilibrium, goods in the same group
  may not share exactly the same price and agents in the same group may not have exactly the same allocation.)


\subsection*{Step 1: Creating Variable Gadgets}

We start with an empty market and create a \emph{variable gadget} 
  for each node $v\in V$ to simulate the variable $x_v$ in the threshold game $H$.
For each node $v \in V$, the variable gadget of $v$ consists of the following three groups of goods and one group of agents: 
\begin{flushleft}\begin{enumerate}
\item Create three groups of goods $G_{v,1},G_{v,2}$ and $G_{v,3}$: 
$G_{v,1}$ has $m^{10}+S_u$ goods, where
$$
S_{u}:= (24m^3 + 12m)\cdot  \outdeg(u) + (24m^3 + 15m)\cdot \indeg(u) -3m,$$
and $G_{v,2}$ and $G_{v,3}$ both have $2m^{10}$ goods.
Let $G_v$ denote the union of $G_{v,1}, G_{v,2}$ and $G_{v,3}$. 
\item Create a group $A_v$ of $5m^{10}$ agents. 
Each agent in  $A_v$ has the following utilities for $G_{v}$:
\begin{equation}\label{utilityutility}
\frac{1}{2m^2-1}
\ \text{for}\ G_{v, 1 },\quad  \frac{m^2+1}{4m^2 -2}\ \text{for}\ G_{v, 2 }, \quad
  1\ \text{for}\ G_{v,3},
\end{equation}
and utility $0$ for every other good in the market (including those created later).
\end{enumerate}\end{flushleft}

Looking ahead, we will prove (in Lemma \ref{lem:variable-price}) that in any $\eps$-approximate HZ equilibrium $(x,p)$ of the final
  HZ market $M_H$, $p(G_{v,1}),p(G_{v,2})$ and $p(G_{v,3})$ must satisfy \vspace{0.15cm}
$$
0\le p(G_{v,1})\lesssim \frac{1}{m^2},\quad p(G_{v,2})\approx \frac{1+p(G_{v,1})}{2}  
\quad\text{and}\quad p(G_{v,3})\approx 2-p(G_{v,1}), \vspace{0.15cm}
$$
 where $p(G_{v,\ell})$ denotes the minimum price of goods in $G_{v,\ell}$.
Indeed, $p(G_{v,1})$ will be used to simulate the variable $x_v$ in the threshold game $H$ and at the end,
  we set $x_v\approx m^2 p(G_{v,1})$ for each $v\in V$ to obtain a
  $\kappa$-approximate equilibrium of $H$.

\subsection*{Step 2: Creating Edge Gadgets}

Next we create an \emph{edge gadget} for each edge $e=(u,v)\in E$ to simulate the action of vertex $u$
  sending a contribution $x_u$ to the summation at vertex $v$ in the threshold game $H$ (see definition \ref{def:threshold}).
For each (directed) edge $e=(u,v)\in E$, the edge gadget of $e$ consists of the following multiple
  groups of goods and agents (for convenience, we only 
  list goods with positive utilities for each group of agents; every other good has utility $0$):
\begin{flushleft}\begin{enumerate}
\item Create a group $G_{e}$ of $32m^5$ goods.
\item Create a group $A_{e,* }$ of $64m^5$ agents. They 
  have utility $1$ for $G_e$. %
\item Create a group of $48m^3$ agents $A_{e,1}$.
 They have utility $1$ for $G_{u, 3}$ and $1/2$ for $G_{v, 1}$. 
\item Create $m$ groups $A_{e,2,\ell}$, $\ell\in [m]$, each of $6$ agents. 
They have $1$ for  $G_{e}$ and $\ell/(2m^3)$ for $G_{v,1}$. 
\item Create $m$ groups $A_{e,3,\ell}$, $\ell\in [m]$, each of $8$ agents. 
They have $1$ for $G_{e}$ and $\ell/(2m^3)$ for $G_{u,1}$. 
\item Create $m$ groups $A_{e,4,\ell}$, $\ell\in [2m]$, each of $18$ agents. They have $1$ for $G_{e}$, $\ell/(2m^3)$ for $G_{v,1}$, $$\frac{1}{4}+ \frac{1}{4m^2} + \frac{1}{m^3}$$ for goods in $G_{u,2}$. 
\end{enumerate}\end{flushleft}
For convenience we write $A_e$ to denote the union of groups $A_{e,1}, A_{e,2,\ell}, A_{e,3,\ell}$ and 
  $A_{e,4,\ell}$, for all $\ell$. 

\subsection*{Step 3: Adding Dummy Goods}

So far we have created
\begin{equation}\label{eq:n}
5m^{10}\cdot |V|+(64m^5+48m^3+50m)\cdot |E|
\end{equation}
many agents and 
$$
\sum_{u\in V} (5 m^{10}+S_u)+32m^5\cdot |E|= (5m^{10}-3m)\cdot |V|+(32m^5+48m^3+27m)\cdot |E| 
$$
many goods.
To finish the construction (since the number of goods needs to match that
  of agents), we create a group of $3m|V|+(32m^5+23m)|E|$ dummy goods,
  which have utility $0$ to every agent in the market.
This finishes the construction of $M_H$ with $n$ agents and goods, where $n$ is given in (\ref{eq:n}).
It is clear that $M_H$ can be built in polynomial time.

Before moving forward,
  we record a list of simple properties about $M_H$:

\begin{fact}\label{basicfact}
The HZ market $M_H$ satisfies the following properties:
\begin{flushleft}\begin{enumerate}
\item Every agent in the market has maximum utility $1$;
\item For each node $v\in V$, 
  the number of agents outside of $A_v$ that have a positive utility on
  at least one group of goods in $G_v$ is at most
  $288m^3+258m=O(m^3)$;
\item For each edge $e\in E$,
  the number of agents outside of $A_{e,^*}$ that have a positive utility on $G_e$
  is $50m=O(m)$.
\end{enumerate}\end{flushleft}
\end{fact}



\subsection{Proof of Correctness}\label{sec:correctness}

Let $\eps=1/n^5$.
We prove three lemmas about variable gadgets in $M_H$ in Section \ref{sec:gadgets}.

We use $p(G_i)$ to denote the minimum price of goods in a group $G_i$.
The first lemma shows that  $p(G_{v,1})$ is between (roughly)
  $0$ and $1/m^2$ and it determines $p(G_{v,2})$ and $p(G_{v,3})$ (approximately). 
\begin{lemma}
	\label{lem:variable-price}
Let $(x,p)$ be an $\eps$-approximate HZ equilibrium of $M_H$. Then $p(G_{v,1})$ satisfies 
$$
0\le p(G_{v,1})\le \frac{1}{m^2}+O\left(\frac{1}{m^6}\right) 
$$
for every $v\in V$.
Moreover, $p(G_{v,2})$ and $p(G_{v,3})$ satisfy\vspace{0.06cm} 
\begin{equation}\label{heheeqeq}
p(G_{v,2}) = \frac{1+p(G_{v,1})}{2}\pm O\left(\frac{1}{m^7}\right)\quad\text{and}\quad 
p(G_{v,3}) = 2 - p(G_{v,1}) \pm O\left(\frac{1}{m^7}\right).
\end{equation}
\end{lemma}

We prove Lemma \ref{lem:variable-price} in Section \ref{sec:gadgets}.

We next show that the variable gadget created for each node $v\in V$ is 
  sensitive to demand from agents outside of $A_v$.
To state the lemma (and the next one), we introduce the following notation:
Let $G^*$ be a subset of goods (which could be a group or the union of multiple groups
  of goods) and $A^*$ be a subset of agents in $M_H$
  (which could be a group or the union of multiple groups).
We let
$$
x^+(G^*,A^*)=\sum_{\substack{i\in A^*\\ j\in G^*:\\ 
u_{i,j}>0}} x_{i,j},
$$ 
i.e., the total allocation of $G^*$ to $A^*$ but limited to those goods 
  in $G^*$ with positive utilities to each agent in $A^*$ only.
We also write $\overline{A}_v$ to denote all agents in $M_H$ outside of $A_v$.

  The second lemma (which we also prove in Section \ref{sec:gadgets}) states that if the total allocation of 
  $G_v$ to agents outside of $A_v$ with positive utilities is
  either more than $S_v+1$ or less than $S_v-1$, 
  then $p(G_{v,1})$ must be at one of the two extreme cases accordingly, i.e., either close to $0$
  or close to $1/m^2$.

\begin{lemma}
	\label{lem:demand}
Let $(x,p)$ be an $\eps$-approximate HZ equilibrium of $M_H$. Then for every $v\in V$:
	\begin{flushleft}\begin{enumerate}
		\item If $x^+(G_v,\overline{A}_v)\ge S_v+1$, then we have
		$$
		p(G_{v,1})=\frac{1}{m^2}\pm O\left(\frac{1}{m^9}\right);
		$$
		\item If $x^+(G_v,\overline{A}_v)\le S_v-1$, then we have  
$p(G_{v,1})\le  O(1/n^2)$.
	\end{enumerate}\end{flushleft}
\end{lemma}

Finally we prove the following lemma about edge gadgets in $M_H$ in Section \ref{sec:edgegadgets}:

\begin{lemma}\label{lem:edgedemand}
Let $(x,p)$ be an $\eps$-approximate HZ equilibrium of $M_H$. For each $e=(u,v)\in E$,
$$x^+(G_u,A_e)= 24m^3+ 12m \pm O(1)\quad\text{and}\quad
x^+(G_v,A_e)=-6m^3p(G_{u,1}) + 24m^3 +15m\pm O(1).$$
\end{lemma}

We now use these lemmas to prove Theorem \ref{theo:ppadhard}:

\begin{proof}[Proof of Theorem \ref{theo:ppadhard} assuming Lemmas \ref{lem:variable-price}, 
  \ref{lem:demand} and \ref{lem:edgedemand}]
Let $H=(V,E)$ be a\ threshold game, and 
  let $(x,p)$ be an $\eps$-approximate HZ equilibrium of $M_H$.
Let $(x_v:v\in V)$ be a profile for $H$ with 
$$x_v=\min\big(1,m^2p(G_{v,1})\big)$$
for each $v\in V$.
We prove below that $(x_v:v\in V)$ is a $\kappa$-approximate equilibrium
  of $H$.

Fix a node $v\in V$. We consider two cases.
\begin{flushleft}\begin{enumerate}
\item Case $1$: $\sum_{u\in N_v} x_u>0.5+\kappa$.
In this case, $x^+(G_v,\overline{A}_v)$ is at most
\begin{align*}
&\outdeg(v)\cdot \big(24m^3+12m+O(1)\big)+\sum_{u\in N_v} \big(24m^3 +15m -6m^3p(G_{u,1}) + O(1)\big)\\
&\hspace{1cm}=S_v+3m-6m^3\sum_{v\in N_v} p(G_{u,1})+O(1)<S_v-1.
\end{align*} 
It follows from Lemma \ref{lem:demand} that $p(G_{v,1})\le O(1/n^2)$ and thus, $x_v\le O(m^2/n^2)< \kappa$.
\item Case $2$: $\sum_{u\in N_v} x_u<0.5-\kappa$. Using $p(G_{u,1})\le 1/m^2+O(1/m^9)$, we have
$$\sum_{u\in N_v} m^2p(G_{u,1})<0.5-\kappa +O(1/m^7).$$
Similarly, $x^+(G_v,\overline{A}_v)$ is at least \begin{align*}
&\outdeg(v)\cdot \big(24m^3+12m- O(1)\big)+\sum_{u\in N_v} \big( 24m^3 +15m- 6m^3p(G_{u,1})- O(1)\big)\\
&\hspace{1cm}=S_v+3m-6m^3\sum_{u\in N_v}p(G_{u,1})-O(1)>S_v+1.
\end{align*} 
It follows from Lemma \ref{lem:demand} that $p(G_{v,1})\ge (1/m^2)-O(1/m^9)$ and thus, $x_v \ge 1-\kappa$.
\end{enumerate}\end{flushleft}
This finishes the proof of the theorem.
\end{proof}

\subsection{Analysis of Variable Gadgets}
\label{sec:gadgets}

We prove Lemma \ref{lem:variable-price} and
  Lemma \ref{lem:demand} in this section. We start with some simple  bounds on prices of
  goods in $G_{e }$ and $ G_{v,3} $, $e\in E$ and $v\in V$:

\begin{lemma}\label{haheha22}
Let $(x,p)$ be an $\eps$-approximate HZ equilibrium of $M_H$.
We have 
  $p(G_e)\ge 2(1-2\eps)$ for every $e\in E$ and $p(G_{v,3})\ge 5/3$ for every $v\in V$.
\end{lemma} 
\begin{proof}
Fix an $e\in E$.
The optimal value of each agent in $A_{e,*}$ is at most $0.5+\eps$;
  otherwise each of them must receive a bundle with value more than $0.5$, which implies
  that each of them gets more than $0.5$ unit of goods in $G_{e}$, contradicting with the fact that 
  there are $64m^5$ many agents in $A_{e,*}$ but only $32m^5$ many goods in $G_e$. 
On the other hand, the optimal value of each agent in $A_{e,*}$ is at least 
  $\min(1,1/p(G_e))$ by Lemma \ref{simplesimplelemma} and thus,
   $p(G_e)\ge 1/(0.5+\eps)>2(1-2\eps)$.
  
Next fix a $v\in V$.
With a similar argument, the optimal value of each agent in $A_v$ is at most 
$$
\frac{1}{5m^{10}}\cdot \left(\frac{ m^{10}+S_v}{2m^2-1}+2m^{10}\cdot\frac{m^2+1}{4m^2-2}+2m^{10}\right)+\eps
< 3/5 
$$
when $m$ is sufficiently large.
On the other hand, by Lemma \ref{simplesimplelemma} the optimal value of each agent in $A_v$ is at least
  $\min(1,1/p(G_{v,3}))$  and thus, $p(G_{v,3})\ge 5/3$.
\end{proof}

From this we can show that every agent in $M_H$ has optimal value at most $0.9$:

\begin{lemma}
Let $(x,p)$ be an $\eps$-approximate HZ equilibrium of $M_H$.
Then every agent in $M_H$ has optimal value (with respect to $p$) at most $0.9$.
\end{lemma}
\begin{proof}
As shown in the previous lemma, the optimal value of each agent in a group $A_v$ of a variable gadget is at most 3/5, and the optimal value of each agent in a group $A_{e,*}$ of an edge gadget is at most $0.5+\epsilon$.
The claim for the agents in the groups $A_{e}$ follows from the
prices of the goods in $G_{u,3}$ and $G_e$, which are the goods that have utility 1 for these agents (the other goods have utility 1/2 or less).
\end{proof}

This allows us to apply lemmas in Section \ref{sec:basic}.
It immediately leads to the following corollary:

\begin{corollary}\label{closecoro}
For every group $G_j$ of goods, the maximum price in $G_j$ is at most
  $p(G_j)+1/n^2$.
\end{corollary}
\begin{proof}
Assume for a contradiction that there is a good in $G_j$ with price at least $p(G_i)+1/n^2$.~Then for each agent $i$ in the market, we have $\alpha^*_i\ge 1/(20n)$ by Lemma \ref{heheclaim1}
  and thus this good is $\Omega(1/n^3)$-suboptimal (by comparing with the 
  good in $G_j$ with price $p(G_j)$). Hence its allocation to agent $i$
  is $O(1/n^2)$, and the total allocation of this good in $x$ is $O(1/n)$, a contradiction.
\end{proof}

Before proving Lemma \ref{lem:variable-price} we show that $p(G_e)$ is very close to $2$:

\begin{lemma}\label{lem:goode}
For every edge $e\in E$ we have $p(G_e)=2\pm O(1/m^4)$.
\end{lemma}
\begin{proof}
We have by Lemma \ref{haheha22} that $p(G_e)\ge 2-O(\eps)$.
For the upper bound note that by Lemma \ref{haheha22} and Lemma \ref{heheclaim1}, goods in 
  $G_e$ are $\Omega(1/n)$-suboptimal to agents with zero utility so their total
  allocation to such agents is $O(n^2\eps)$ by Lemma \ref{fact:subopt-goods}.
By Fact \ref{basicfact} the total allocation of $G_e$ to agents outside $A_{e,*}$ with a positive utility is $O(m)$ and thus,
  the rest of $32m^5-O(m)$ units of $G_e$ are allocated to agents in $A_{e,*}$.
So 
$$
\left(32m^5-O(m)\right)p(G_e)\le 64m^5(1+\eps),
$$
which implies that $p(G_e)\le 2+O(1/m^4)$. This finishes the proof of the lemma.
\end{proof}

We are now ready to prove Lemma \ref{lem:variable-price}:
\begin{figure}[!ht]
	\centering
	\includegraphics[width=0.8\textwidth]{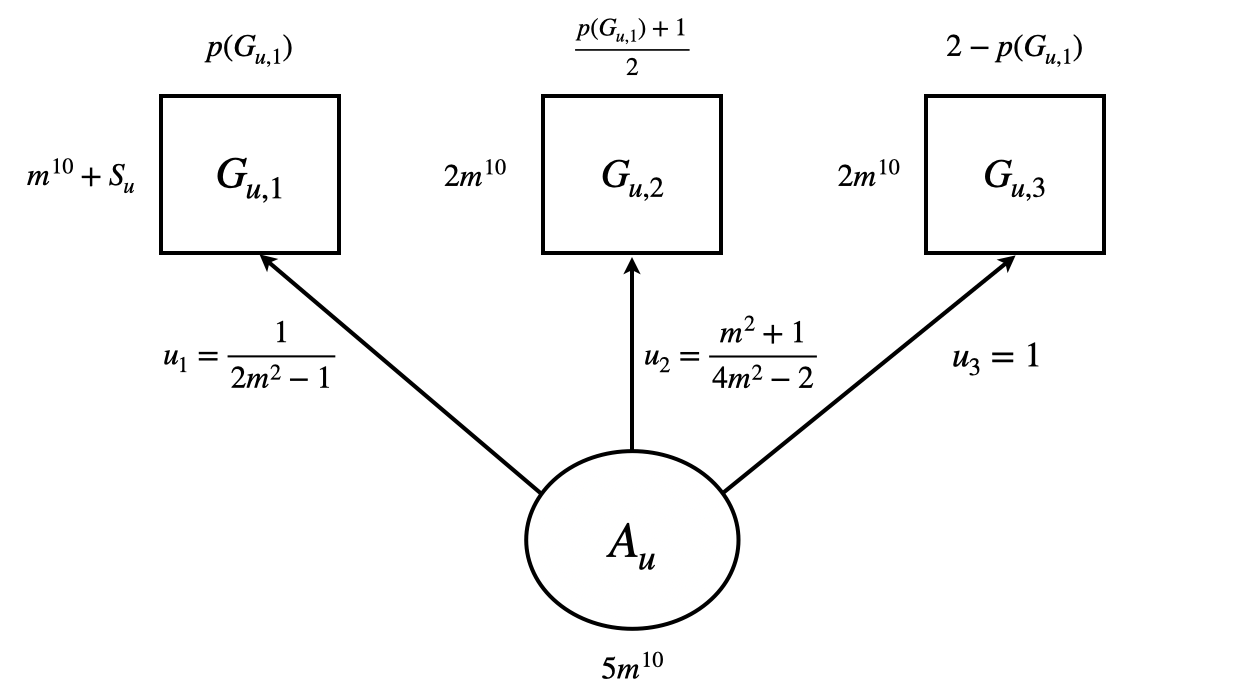}
  \vspace{0.2cm}\caption{The variable gadget.}
\end{figure}
\begin{proof}[Proof of Lemma \ref{lem:variable-price}]
Fixing any node $v\in V$, we let $q_\ell$ denote $p(G_{v,\ell})$ and $y_\ell$ to denote the 
  total allocation of $G_{v,\ell}$ to agents in $A_v$ in $x$, for each $\ell\in \{1,2,3\}$.
  We also write $u_\ell$ to denote the utility of $G_{v,\ell}$ to agents in $A_{v}$
    given in (\ref{utilityutility}).
We start by showing that most goods in $G_v$ go to $A_v$.

\begin{claim}\label{claim:recall}
We have $y_1\ge m^{10}-O(m^3)$ and $y_2,y_3\ge 2m^{10}-O(m^3)$. 
\end{claim}
\begin{proof}
Let $\alpha^*$ and $\mu^*$ be an optimal solution to the dual LP of agents in $A_v$.
Then $\alpha^* q_\ell+\mu^*\ge u_\ell$ for each $\ell$.
We consider the following two cases. 

First we consider the case when $\mu^*\ge u_1/2=\Omega(1/m^2)$.
This implies that goods outside of $G_v$ are  $\Omega(1/m^2)$-suboptimal for $A_v$ and thus, by Lemma \ref{fact:subopt-goods},
the total allocation of them to agents in $A_v$ is $O(m^{10})\cdot O(m^2 \eps)<1$ .
As a result,  
 $y_1+y_2+y_3\ge 5m^{10}-1$ from which the claim follows. 

Next consider the case when $\mu^*<u_1/2$. By Lemma \ref{simplelemma1} ($\alpha^*\le 1$) we have
  $q_\ell\ge \Omega(1/m^2)$ for every~$\ell$.
This implies that agents with zero utilities to $G_v$ can be allocated only
  $n\cdot O(n m^2 \eps) < 1$ units of $G_v$ given that they are $\Omega(1/nm^2)$-suboptimal by Lemma \ref{heheclaim1}.
On the other hand, by Fact \ref{basicfact} the allocation to agents outside $A_v$ with positive utilities for
  $G_v$ is at most $O(m^3)$.
So all the rest of $G_v$ must be allocated to $A_v$ and the claim follows. %
\end{proof}

	Now that we have $y_\ell\ge m^{10}-O(m^3)$ for all $\ell\in \{1,2,3\}$, we proceed to prove (\ref{heheeqeq}). 
	Let $(\alpha^*,\mu^*)$ denote an optimal solution to the dual LP for $A_v$.
By Lemma~\ref{fact:subopt-goods} and taking $\delta = 20\eps$,
  we have 
\begin{equation}\label{blablabla}
u_\ell\le \alpha^*q_\ell +\mu^*\le u_\ell+\delta,\quad\text{for all $\ell\in \{1,2,3\}$.}
\end{equation}  
If this were not the case (i.e. the second inequality is violated for some $\ell$), then goods in $G_{v,\ell}$ are
  $\delta$-suboptimal to $A_v$ and their total allocation to agents in $A_v$
  can be no more than 
$
5m^{10}\cdot 2\eps/\delta=m^{10}/2,
$
a contradiction.

Combining (\ref{blablabla}) and $u_2=(3u_1+u_3)/4$, we have
$$
\alpha^*\left(\frac{3q_1+q_3}{4}\right)+\mu^*-\delta\le \alpha^*q_2+\mu^*
\le \alpha^*\left(\frac{3q_1+q_3}{4}\right)+\mu^*+\delta.
$$
Using $\alpha^*\ge 1/(20n)$ from Lemma \ref{heheclaim1}, we have
	\begin{align}\label{hehehehehehe1}
	q_2 = \frac{3q_1 + q_3}{4} \pm O(n\eps) . 
	\end{align}
	
	Next, using Corollary \ref{secondcoro} and Corollary \ref{closecoro} we have
$$
5m^{10}(1-O(1/n^2))\le q_1y_1+q_2y_2+q_3y_3\le 5m^{10}(1+\eps+1/n^2).
$$	
Plugging in $y_1=m^{10}\pm O(m^3)$ and $y_2,y_3=2m^{10}\pm O(m^3)$ and (\ref{hehehehehehe1}), we have
$ 
q_1+q_3=2\pm O(1/m^7).
$ 
Together with (\ref{hehehehehehe1}) again we obtain
%
$$q_3=2-q_1\pm O(1/m^7)\quad\text{and}\quad q_2=(1+q_1)/2\pm O(1/m^7).$$
	
	Finally we give an upper bound on $q_1$. We first note that $q_1<q_3$; otherwise
	  goods in $G_{v,1}$ are $\Omega(1)$-suboptimal to agents in $A_v$, contradicting with $y_1=m^{10}\pm O(m^3)$.
	  Using (\ref{blablabla}) we have
	$$
	\mu^*\le \frac{u_1q_3-u_3q_1+ O(\delta)}{q_3-q_1}.
	$$
But when $q_1\ge 1/m^2+1/m^6$ (and thus, $q_3\le 2-1/m^2 $), the nominator of the RHS is
$$
u_1q_3-u_3q_1\le \frac{1}{2m^2-1}\cdot \left(2-\frac{1}{m^2}\right)-1\cdot \left(\frac{1}{m^2}+\frac{1}{m^6}\right)\le -\frac{1}{m^6}.
$$
So we have $\mu^*<0$, in contradiction with $\mu^*\ge 0$ by Lemma \ref{simplelemma1}.
\end{proof}

Next we prove Lemma \ref{lem:demand}:

\begin{proof}[Proof of Lemma \ref{lem:demand}]
We use the same notation from the proof of the last lemma.

First given that $q_2$ and $q_3$ are $\Omega(1)$, the total allocation of
  $G_{v,2}$ to agents with zero utility on $G_{v,2}$ is at most $n\cdot O(n\eps)$ using Lemma \ref{fact:subopt-goods}; 
  the same applies to $G_{v,3}$.

  Suppose that $x^+(G_v,\overline{A}_v)\le S_v-1$.
Because $G_v$ contains $5m^{10}+S_v$ goods while $A_v$ contains only $5m^{10}$ agents, 
  for $G_v$ to be fully sold out,
   the total allocation of $G_{v,1}$ to agents with zero utility 
  on $G_{v,1}$ must be $1-o_n(1)$. This implies that $q_1\le 1/n^2$ since otherwise,
  the total allocation for $G_{v,1}$ is at most $n\cdot O(n^3\eps)=o_n(1)$, using $\eps=1/n^5$.


Next, suppose $x^+(G_v,\overline{A}_v)\ge S_v+1$. 
Given that there are $5m^{10}+S_v$ goods in $G_v$ and $5m^{10}$ agents in $A_v$,
  there must be an agent in $A_v$ who is allocated 
  at least $1/(5m^{10})$-unit of goods outside of $G_v$ (for which it has zero utility).
Since such goods are $\mu^*$-suboptimal,
  we have $\mu^*\le 5m^{10}\eps$.
On the other hand, recall (\ref{blablabla}) with 
   $\delta=20\eps$.
We have $\alpha^* (q_1+q_3)+2\mu^*=u_1+u_3\pm 2\delta$ and thus,
  $$\alpha^*=\frac{u_1+u_3}{2}\left(1\pm O\left(\frac{1}{m^7}\right)\right)$$
  using Lemma \ref{lem:variable-price}.
Then $q_1=(u_1-\mu^*\pm \delta)/\alpha^* = 1/m^2 \pm O(1/m^{9})$. 
\end{proof}

\subsection{Analysis of Edge Gadgets}\label{sec:edgegadgets}

In this subsection we prove Lemma \ref{lem:edgedemand}.
Let $(x,p)$ be an $\eps$-approximate HZ equilibrium of $M_H$~and $e=(u,v)\in E$ be an edge in $H$.
We work on agents in 
  $A_{e,1},A_{e,2,\ell},A_{e,3,\ell}$ and $A_{e,4,\ell}$ to understand their allocations of goods
  with non-zero utilitites. 
We start with agents in $A_{e,1}$:




 



\begin{figure}[!ht]
	\centering
	\includegraphics[width=0.8\textwidth]{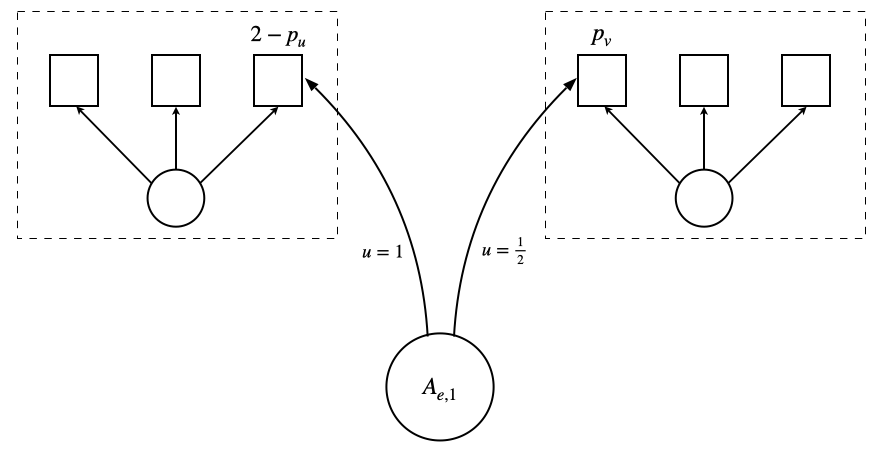}
	 \vspace{0.15cm}\caption{Edge gadget: Agents $A_{e,1}$}
\end{figure}




\begin{lemma}
	\label{lem:gadget1}
The allocation of $G_{u,3}$ and $G_{v,1}$ to each agent in $A_{e,1}$ is given by  
$$\frac{1}{2} + \frac{p(G_{u,1}) - p(G_{v,1})}{4}+ O\left(\frac{1}{m^4}\right)\quad\text{and}\quad
 \frac{1}{2} + \frac{p(G_{v,1}) - p(G_{u,1})}{4} + O\left(\frac{1}{m^4}\right),$$
 respectively. 
\end{lemma}
\begin{proof}
The dual LP for each agent in $A_{e,1}$ is to minimize $\alpha+\mu$ subject to the following constraints: $\alpha\ge 0$;
$\alpha p_j + \mu \geq 1$ for $j \in G_{u,3}$;
$\alpha p_j + \mu \geq 1/2$ for   $j \in G_{v,1}$;
and $\alpha p_j + \mu \geq 0$ for  $j \notin G_{u,3} \cup G_{v,1}$.
The constraints for the minimum-priced goods in the groups
$G_{u,3}$, $G_{v,1}$
dominate the constraints for the others goods in the groups
(since $\alpha \geq 0$),
hence these constraints 
are equivalent to 
\begin{equation}\label{hah}
\alpha\cdot p(G_{u,3})+\mu\ge 1\quad\text{and}\quad 
\alpha\cdot p(G_{v,1})+\mu \ge 1/2.
\end{equation}
The constraint for the good $j$ with price $0$ yields $\mu \geq 0$,
and this subsumes the constraints for all $j \notin G_{u,3} \cup G_{v,1}$.
Thus, the dual LP is to minimize $\alpha+\mu$ subject to $\alpha\ge 0$, $\mu \geq 0$, and (\ref{hah}) above.
The optimal solution $\alpha^*, \mu^*$ satisfies the (\ref{hah}) as equalities.
Thus, solving the dual LP we get 
$$
\alpha^*=\frac{1}{2(p(G_{u,3})-p(G_{v,1}))}\approx \frac{1}{4}\quad\text{and}\quad \mu^*=\frac{p(G_{u,3})-2p(G_{v,1})}{2(p(G_{u,3})-p(G_{v,1}))}\approx\frac{1}{2}.
$$
Consider an agent in $A_{e,1}$ and write
$x_u$ and $x_v$ respectively to denote her allocation of goods from $G_{u,3}$ and $G_{v,1}$.
Since every good that is not in $G_{u,3} \cup G_{v,1}$
  is $\mu^*$-suboptimal for the agent, 
  it follows from Corollary \ref{maincoro} and Corollary \ref{closecoro} that  
$$
x_u+x_v=1\pm O(\eps)\quad\text{and}\quad p(G_{u,3})x_u+p(G_{u,1})x_v=1\pm O(1/n^2).
$$
From Lemma \ref{lem:variable-price},
$p(G_{u,3})=2-p(G_{u,1})\pm O(1/m^7)$.
For ease of notation, let us use $p_v $ to denote $p(G_{v,1})$
and $p_u$ for $p(G_{u,1})$. 
Solving the above equations for $x_u$ and $x_v$  gives us
	\begin{align*}
	x_v = \frac{1 - p_u}{2 - p_u - p_v} \pm O\left(\frac{1}{m^7}\right) \quad \text{and} \quad
	x_u = \frac{1 - p_v}{2 - p_u - p_v} \pm O\left(\frac{1}{m^7}\right).
	\end{align*}
	Since $p_u, p_v \in [0, 1/m^2 \pm O(1/m^6)]$, by performing the first order approximation, we have
	\begin{align*}
 \frac{1 - p_u}{2 - p_u - p_v} 
	=  \frac{(1 - p_u)(2+p_u + p_v)}{(2 - p_u - p_v)(2+p_u+p_v)}  
	=  \frac{2 + p_v - p_u - p_u(p_u+p_v)}{4 - (p_u + p_v)^2}  
	= \frac{2+p_v - p_u}{4} \pm O\left(\frac{1}{m^4}\right).
	\end{align*}
The result for $x_u$ can be shown similarly.
\end{proof}

 

Now we work on agents in $A_{e,2,\ell}$ for each $\ell\in [m]$:

\begin{figure}[!ht]
	\centering
	\includegraphics[width=\textwidth]{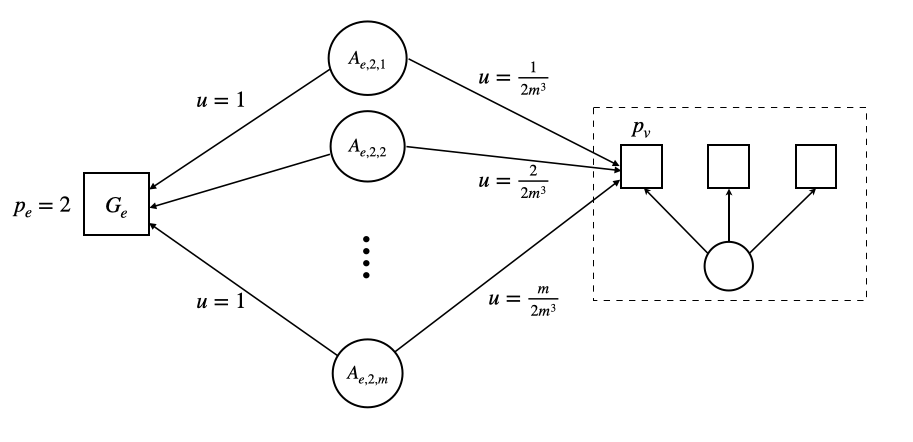}
  \vspace{-0.8cm}\caption{Edge gadget: Agents $A_{e,2,\ell}$ with $\ell\in [m]$.}
\end{figure}




\begin{lemma}
	\label{lem:gadget2}
For each $\ell\in [m]$, the allocation of goods in $G_{v,1}$ 
  to each agent in $A_{e,2,\ell}$ is $O( m^4\eps )$ if $\ell/m^3\le p(G_{v,1})-1/m^4$ and 
  $0.5 \pm O(1/{m^2})$
  if $\ell/m^3\ge p(G_{v,1})+ 1/m^4$.
\end{lemma}
\begin{proof}
The dual LP of each agent in $A_{e,2,\ell}$ is to minimize $\alpha+\mu$ subject to the following constraints: 
$\alpha\ge 0$; $\mu\ge 0$;
$\alpha\cdot p(G_e)+\mu\ge 1$; 
$\alpha\cdot p(G_{v,1})+\mu\ge \ell/2m^3$.
(We have simplified the dual LP following similar arguments used at the beginning of the proof for the previous lemma.)
The optimal solution  $(\alpha^*,\mu^*)$ now has two cases.

If $p(G_{v,1})\ge \ell/m^3+1/m^4$, 
  the optimal solution $(\alpha^*,\mu^*)$ is 
  $\alpha^*=1/p(G_e)$ and $\mu^*=0$. Thus,  
$$
\alpha^*\cdot p(G_{v,1})+\mu^*=\frac{p(G_{v,1})}{p(G_e)} \ge \frac{\ell}{2m^3}+\Omega\left(\frac{1}{m^4}\right)   
$$
using Lemma \ref{lem:goode}.
So goods in $G_{v,1}$ are $\Omega(1/m^4)$-suboptimal to the agent and it follows directly from
  Lemma \ref{fact:subopt-goods} that the agent is allocated at most $O(m^4 \eps)$ units of $G_{v,1}$.

If $p(G_{v,1})< \ell/m^3-1/m^4$,
  the optimal solution is given by
$$
\alpha^*=\frac{1-\ell/(2m^3)}{p(G_e)-p(G_{v,1})}\approx \frac{1}{2}\quad\text{and}\quad 
\mu^*=\frac{(\ell/2m^3)p(G_e)-p(G_{v,1})}{p(G_e)-p(G_{v,1})}\ge \Omega\left(\frac{1}{m^4}\right).
$$
So goods not in $G_e \cup G_{v,1}$ are $\mu^*$-suboptimal.
Let $x_e$ and $x_v$ be respectively the allocation of goods in $G_e$ and $G_{v,1}$ to the agent.
It follows from Corollary \ref{maincoro} and Corollary \ref{closecoro} that 
\begin{align*}
x_e+x_v&=1\pm O(m^4\eps)\quad\text{and}\quad
p(G_e)x_e+p(G_{v,1})x_v =1\pm O(1/n^2).
\end{align*}
	Solving the above equations, and recalling  
	$p(G_e) = 2 \pm O(1/m^4)$ and $p(G_{v,1}) = O(1/m^2)$, we have 
	\begin{align*}
	x_v = \frac{p(G_e)-1}{p(G_e)-p(G_{v,1})} \pm O\left(\frac{1}{n^2}\right) = \frac{2+p(G_{v,1})}{4}  \pm O\left(\frac{1}{m^{4}}\right). 
	\end{align*}
	This finishes the proof of the lemma.
\end{proof}


\begin{lemma}
	\label{lem:gadget2-com}
The total allocation of goods in $G_{v,1}$ to all
  agents in $A_{e,2,\ell}$, $\ell\in [m]$, is $$3m\big(1 - m^2 p(G_{v,1})\big) \pm O(1).$$
\end{lemma}
\begin{proof}
Let $\ell^*=\lfloor m^3 p(G_{v,1})\rfloor\in [0:m]$. 
Then for each $\ell \in [m]$ with $\ell\le \ell^*-1$, the allocation of $G_{v,1}$ to each agent in $A_{e,2,\ell}$ is 
  $O(m^4\eps)$;
for each $\ell\in [m]$ with $\ell\ge \ell^*+2$, the allocation is $0.5 \pm O(1/m^2)$.
For $\ell^*$ and $\ell^*+1$, the allocation is between $0$ and $1$ trivially.

Since there are 6 agents in each group $A_{e,2,\ell}$, the total allocation to all agents in all groups is 
$$
6\sum_{\ell\in [\ell^*-1]} O(m^4\eps)+6\sum_{\ell\in [\ell^*+2:m]} \left(\frac{1}{2}\pm O\left(\frac{1}{m^2}\right)\right)+ O(1)
=6\sum_{\ell\in [\ell^*+2:m]} \frac{1}{2} \pm O(1).
$$
The number of summands in the last expression is 
$$\max(0,m-\ell^*-1)= 
m \left( 1-m^2 p(G_{v,1})\right)\pm O(1) $$
and the lemma follows.
\end{proof}

The following lemma for agents in $A_{e,3,\ell}$ can be proved similarly:

\begin{lemma}\label{lem:similar}
The allocation of goods in $G_{u,1}$ to all
  agents in $A_{e,3,\ell}$, $\ell\in [m]$, is $$4m\big(1 - m^2 p(G_{u,1})\big) \pm O(1).$$
\end{lemma}

Finally we work on agents in $A_{e,4,\ell}$, $\ell\in [2m]$:

\begin{figure}[!ht]
	\centering
	\includegraphics[width=\textwidth]{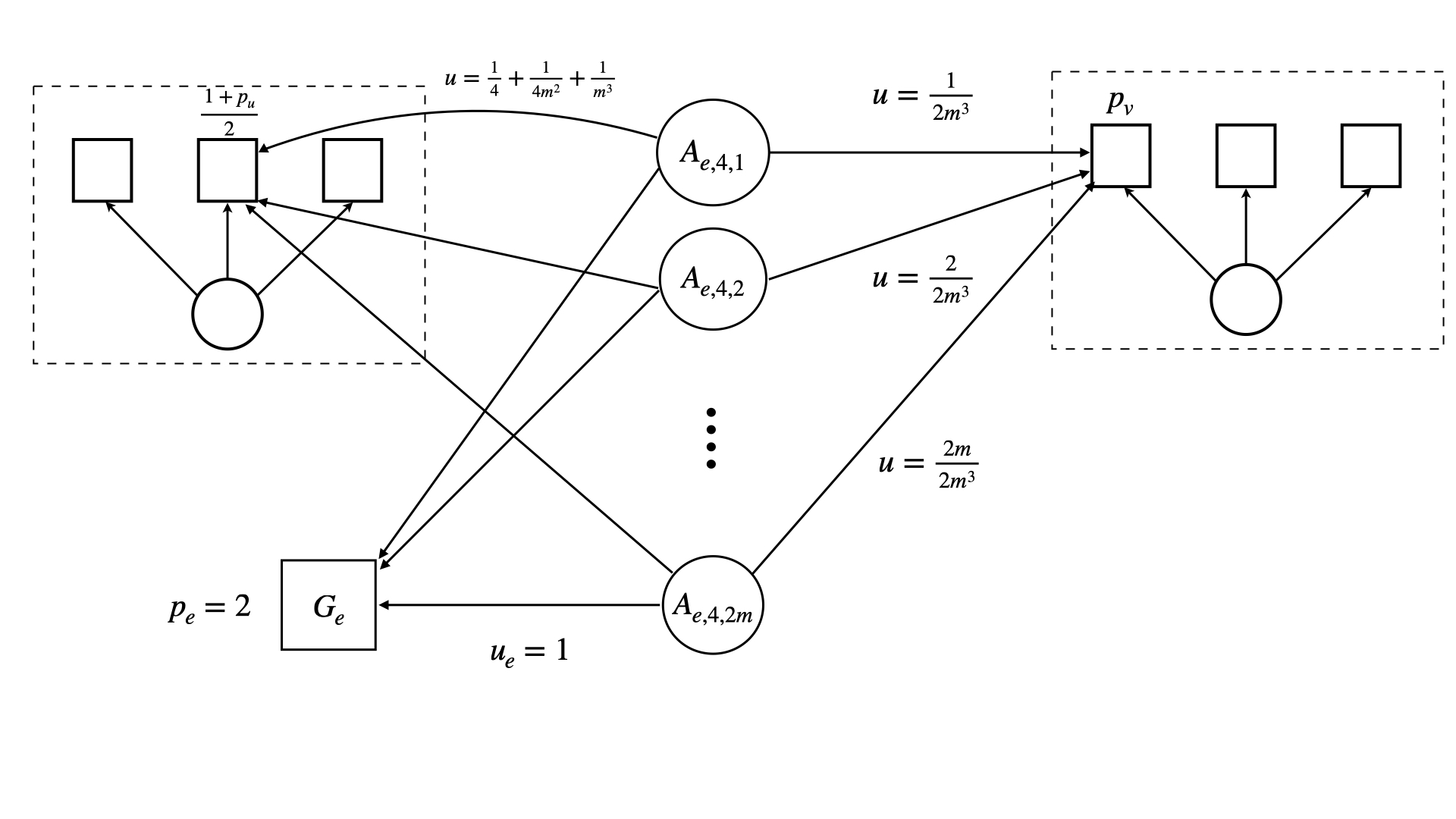}
  \vspace{-1.8cm}\caption{Edge gadget: Agents $A_{e,4,\ell}$ with $\ell\in [2m]$.}
\end{figure}

\begin{lemma}
	\label{lem:gadget3}
For each $\ell\in [2m]$, we have the following two cases for an agent in $A_{e,4,\ell}$:
\begin{flushleft}\begin{enumerate}
\item The allocation of $G_{u,2}$ is $O(m^3\eps)$ and the allocation of $G_{v,1}$ is 
  $0.5\pm O(1/m^2)$, if $\ell$ satisfies
\begin{equation}\label{eq:case1}
\frac{\ell}{2m^3}\ge \frac{p(G_{v,1})}{2}-\frac{p(G_{u,1})}{3}+\frac{1}{3m^2}+\frac{2}{m^3}.
\end{equation}
\item The allocation of $G_{v,1}$ is $O(m^3\eps)$ and the allocation of $G_{u,2}$ is
  $2/3\pm O(1/m^2)$, if $\ell$ satisfies
\begin{equation}\label{eq:case2}
\frac{\ell}{2m^3}\le \frac{p(G_{v,1})}{2}-\frac{p(G_{u,1})}{3}+\frac{1}{3m^2}.
\end{equation}
\end{enumerate}\end{flushleft}
\end{lemma}
\begin{proof}
The dual LP for an agent in $A_{e,4,\ell}$ is to minimize $\alpha +\mu $ subject to $\alpha \ge 0$, $\mu \ge 0$,
\begin{equation}\label{eq:lp}
\alpha \cdot p(G_e)+\mu \ge 1,\quad
\alpha \cdot p(G_{u,2})+\mu \ge \frac{1}{4}+\frac{1}{4m^2}+\frac{1}{m^3}\quad\text{and}\quad
\alpha \cdot p(G_{v,1})+\mu \ge \frac{\ell}{2m^3}.
\end{equation}

Geometrically, the feasible space of the LP is a region of
the $\alpha -\mu$ plane, bounded from below by a a piecewise linear
convex curve whose segments correspond to (some of)
the above constraints. The optimum is achieved at a vertex 
of the boundary curve, where the constraints corresponding to the two adjacent segments are tight.
We consider the two cases.

First we consider the case of (\ref{eq:case1}).
In this case, the first and third inequalities of (\ref{eq:lp}) are tight at optimality,
and the optimal solution $(\alpha^*,\mu^*)$ is 
$$
\alpha^*=\frac{1-(\ell/2m^3)}{p(G_e)-p(G_{v,1})}\quad\text{and}\quad
\mu^*=\frac{(\ell/2m^3)\cdot p(G_e)-p(G_{v,1})}{p(G_e)-p(G_{v,1})}.
$$
Using (\ref{eq:case1}) and $p(G_e) \approx 2$, 
$p(G_{u,1}) \leq 1/m^2 \pm O(1/m^6)$, it is easy to check
that $\mu^* = \Omega(1/m^3)$.
Therefore, all goods that are not in $G_e \cup G_{v,1} \cup G_{u,2}$
are $\Omega(1/m^3)$-suboptimal.
Furthermore, also all goods in $G_{u,2}$ are  $\Omega(1/m^3)$-suboptimal, i.e.,
$$
\alpha^*\cdot p(G_{u,2})+\mu^*-\left(\frac{1}{4}+\frac{1}{4m^2}+\frac{1}{m^3}\right) \ge \Omega\left(\frac{1}{m^3}\right),
$$
because multiplying the left-hand-side by $p(G_e)-p(G_{v,1})\approx 2$ we get
\begin{align*}
&\frac{1}{2}-\frac{\ell}{4m^3}+\frac{p(G_{u,1})}{2}+\frac{\ell}{m^3}-p(G_{v,1})
-\frac{1}{2}-\frac{1}{2m^2}-\frac{2}{m^3}+\frac{p(G_{v,1})}{4}\pm O\left(\frac{1}{m^4}\right)\\[0.8ex]
&\hspace{0.5cm}=\frac{3\ell}{4m^3}+\frac{p(G_{u,1})}{2}-\frac{3p(G_{v,1})}{4}-\frac{1}{2m^2}-\frac{2}{m^3}\pm O\left(\frac{1}{m^4}\right)
\ge \Omega\left(\frac{1}{m^3}\right).
\end{align*}
So it follows from Lemma \ref{fact:subopt-goods} that the total allocation to the agent (and the cost) of all goods that are not
in $G_e$ and $G_{v,1}$ is $O(m^3\eps)$.

Using $x_e$ and $x_v$ to denote the allocation of goods in $G_e$ and $G_{v,1}$ to the agent respectively,
  it follows from Corollary \ref{maincoro} and Corollary \ref{closecoro} that
	\begin{align*}
		x_e + x_v &= 1 \pm O(m^3\eps ) \quad\text{and}\quad
		p(G_{v,1}) x_v + p(G_{e})x_e  = 1 \pm O(1/n^2)
	\end{align*}
	and using $p(G_e) \approx 2$ and $p(G_{v,1}) = O(m^2)$,
	one can derive that $x_v = 0.5 \pm O(1/m^2)$.

	Similarly, one can show that when (\ref{eq:case2}) holds,
then the first two inequalities of (\ref{eq:lp}) are tight at the optimal solution
$\alpha^*, \mu^*$. We have again $\mu^* = \Omega(1/m^3)$,
and thus all goods outside $G_e \cup G_{u,2} \cup G_{v,1}$ are $\Omega(1/m^3)$-suboptimal.
Furthermore, in this case all goods in $G_{v,1}$ are
also $\Omega(1/m^3)$-suboptimal. Hence the total allocation
(and cost) of these goods to the agent is  $O(m^3\eps)$.
We can set up similarly as in the previous case 
the equations for the allocations
$x_e, x_u$ of the goods in $G_e$, $G_{u,2}$ respectively to the agent,
and solve them to get $x_u = 2/3 \pm O(1/m^2)$.
This concludes the proof.
\end{proof}


\begin{lemma}
	\label{lem:gadget3-com}
The total allocation of goods in $G_{v,1}$ to agents in $A_{e,4,\ell}$, $\ell\in [2m]$, is  
	\[
	18m\left(\frac{2}{3} + \frac{m^2}{3}\cdot p(G_{u,1}) - \frac{m^2}{2}\cdot p(G_{v,1})\right) + O(1),
	\]
	 and the total allocation of goods in  $G_{u,2}$ to these agents is 
	 \[
	 18m\left(\frac{4}{9} - \frac{4m^2}{9}\cdot p(G_{u,1}) + \frac{2m^2}{3}\cdot p(G_{v,1})\right) + O(1).
	 \]
\end{lemma}

\begin{proof}
Let $\ell^*$ be chosen as
$$
\ell^*=\left\lfloor m\left(\frac{2 }{3}+m^2\cdot p(G_{v,1})-
\frac{2m^2 }{3}\cdot p(G_{u,1})\right)\right\rfloor \le \frac{5m}{3}.
$$
We start with goods in $G_{v,1}$. By Lemma~\ref{lem:gadget3}, the total allocation is $18$ times
$$
\sum_{\ell\in [\ell^*]} O(m^3\eps)+\sum_{\ell\in [\ell^*+6:2m]} 
  \left(0.5\pm O\left(\frac{1}{m^2}\right)\right)+ O(1)
  =\sum_{\ell\in [\ell^*+6:2m]} 
   0.5 \pm O(1).
$$
The number of terms in the last sum is 
$$
2m-m\left(\frac{2 }{3}+m^2\cdot p(G_{v,1})-
\frac{2m^2}{3}\cdot p(G_{u,1})\right)
\pm O(1)
$$
It follows that the total allocation is 
$$
18m\left(\frac{2}{3}+\frac{ m^2}{3}\cdot p(G_{u,1})-\frac{m^2}{2}\cdot p(G_{v,1})\right) \pm O(1).
$$

	
	Similarly, the total allocation of goods $G_{u,2}$ is 18 times
	$$
	\sum_{\ell\in [\ell^*]} \left(\frac{2}{3}\pm O\left(\frac{1}{m^2}\right)\right)
	+\sum_{\ell\in [\ell^*+6:2m]} O(m^3\eps)+ O(1)=\sum_{\ell\in [\ell^*]} \frac{2}{3} \pm O(1).
	$$
	The number of terms in the last sum is 
	$$
	m\left(\frac{2 }{3}+m^2\cdot p(G_{v,1})-
\frac{2m^2}{3} \cdot p(G_{u,1})\right)\pm O(1)
	$$
and the lemma follows.
\end{proof}

We are now ready to prove Lemma \ref{lem:edgedemand}:

\begin{proof}[Proof of Lemma \ref{lem:edgedemand}]
Combining Lemma~\ref{lem:gadget1}, Lemma~\ref{lem:gadget2-com}, Lemma \ref{lem:similar} and Lemma~\ref{lem:gadget3-com}, we have the total allocation of nonzero-utility goods in $G_v$ to agents in $A_e$ is 
\begin{align*}
&~48m^3 \left(\frac{1}{2} + \frac{p(G_{v,1})- p(G_{u,1})}{4}\right) + 18m\left( 
\frac{2}{3} + \frac{m^2}{3}\cdot p(G_{u,1}) - \frac{m^2}{2}
\cdot p(G_{v,1})\right) \\[0.6ex]
&~\hspace{3cm}+ 3m\big(1-m^2\cdot  p(G_{v,1})\big) \pm O(1)  \notag\\[0.8ex]
&\hspace{1cm}= -6m^3\cdot p(G_{u,1}) + 24m^3 +15m\pm O(1). 
\end{align*}
The total allocation of nonzero-utility goods in $G_u$ to $A_e$ is 
\begin{align}
&~48m^3 \left(\frac{1}{2} + \frac{p(G_{u,1}) - p(G_{v,1})}{4}\right) + 18m 
  \left(\frac{4}{9} - \frac{4m^2}{9}\cdot p(G_{u,1}) + \frac{2m^2}{3}\cdot p(G_{v,1})\right) \notag\\[0.6ex]
&~\hspace{4cm}+ 4m \big(1 - m^2\cdot p(G_{u,1})\big) \pm O(1)   \notag\\[0.8ex]
&\hspace{1cm}= 24m^3+ 12m \pm O(1). 
\end{align}
This finishes the proof of the lemma.
\end{proof}

\section{Hardness of Approximating Optimal Social Welfare}
\label{sec:np}

In this section we study the problem of 
  approximating the optimal social welfare (defined as the total utility of all agents)
  achievable by an HZ equilibrium.
For this purpose we study the following gap problem for a constant $\rho<1$:
  the input is an HZ market $M$ together with a parameter $\textsf{SW}$, and it is promised
  that the optimal social welfare achievable by an exact HZ equilibrium of $M$ is either at least $\textsf{SW}$
  or at most $\rho\cdot \textsf{SW}$.
The goal is to tell which case it is.
We show that there is no polynomial-time algorithm
  for the gap problem when $\rho>175/176$, assuming $\mathsf{NP}\ne \mathsf{P}$.

\begin{theorem}
	\label{thm:np-hard}
	Assuming $\mathsf{NP}\ne \mathsf{P}$, for any constant $\eps > 0$, there is no polynomial-time algorithm 
	for the gap problem when $\rho=({175}/{176}) + \eps$.
\end{theorem}

\subsection{Construction}
We reduce from $\MAX~3\SAT$, which is hard to approximate better than ${7}/{8}$ \cite{haastad2001some}: Given a $3\SAT$ instance, it is NP-hard to distinguish the case that the formula is satisfiable from the
case that every truth assignment satisfies at most a fraction $\frac{7}{8}+\eps$ of the clauses, for any $\eps>0$.
Given a $3\SAT$ instance with $m$ clauses and $n$ variables, we construct the following HZ market. Throughout the proof, we fix $K = m^3$.

\paragraph{Creating Variable Gadget} We first introduce the variable gadget. For convenience, we only list non zero utilities. 
For each $i \in [n]$
\begin{enumerate}
	\item Create three groups of goods $G_{i, 1}$, $G_{i, 2}$, $G_{i,3}$, and $|G_{i,1}| = K$, $|G_{i,2}| = 2K$ and $|G_{i,3}| = K$.
	\item Create two groups of agents $A_{i,1}$, $A_{i,2}$, and $|A_{i,1}| = |A_{i,2}| = 2K$. 
	\item Agents in $A_{i,1}$ have utility $\frac{1}{2K^2}$ for $G_{i,1}$,  $\frac{1}{K^2}$ for $G_{i,2}$.  Agents in $A_{i,2}$ have utility $\frac{1}{2K^2}$ for $G_{i,3}$, $\frac{1}{K^2}$ for $G_{i,2}$.
\end{enumerate}

In an (exact) HZ equilibrium, all goods within a group have the same price. We use $p(G_{i, \ell})$ to denote the price, $\ell \in [3]$.

\paragraph{Creating Clause Gadget}

We next construct clause gadgets. For each $j \in [m]$, 
\begin{enumerate}
	\item Create a group $G_j$ of $K$ goods
	\item Create a group $A_{j, *}$ of $2K$ agents, who have utility $1/K^2$ for $G_{j}$.
	\item Create an agent $A_{j}$ with utility $1$ for $G_j$. It has utility $5/6$ for $G_{i, 1}$ if the $j$-th clause contains $x_i$ and utility $5/6$ for $G_{i, 3}$ if the $j$-th clause contains $\bar{x}_i$.
\end{enumerate}

\paragraph{Adding Dummy Goods}
Thus far, we have described $4Kn + Km$ goods and $4Kn + (2K+1)m$ agents. We add $(K +1)m$ extra dummy goods that have zero utilities for all agents. In a normalized (exact) HZ equilibrium, these goods have zero price.

\subsection{Proof of Correctness}

We provide the proof of completeness and soundness separately.

\paragraph{Completeness} Given a $3\SAT$ instance that has a satisfying assignment, we construct a HZ equilibrium with social welfare at least ${11}m/{12} - O(1/m^2)$. Fix a satisfying assignment.

We assign the $j$-th clause to the $\phi(j)$-th variable, if the latter satisfies the clause. If there are multiple such variables, we choose an arbitrary one. 
We set $\ell(j) = 1$ if the $j$-th clause contains $x_{\phi(j)}$, otherwise $\ell(j) = 3$. 
Let $s_i$ be the total number of clauses assigned to the $i$-th variable.
The equilibrium prices are as follows. 
\begin{enumerate}
	\item The price of dummy goods is $0$.
	\item The price of $G_j$ is $p(G_j) = \frac{2K+1}{K}$, $j \in [m]$.
	\item For variable gadget $i \in [n]$, if $x_i = 1$, then $(p(G_{i, 1}), p(G_{i, 3}), p(G_{i, 2})) = (0, \frac{8}{5}, \frac{4}{5})$,  otherwise, we have $(p(G_{i, 1}), p(G_{i, 2}), p(G_{i, 3})) = (\frac{4}{5}, \frac{8}{5}, 0)$. 
\end{enumerate}
Next, we specify the equilibrium allocation.
\begin{enumerate}
	\item Agents of $A_{j, *}$ take $\frac{2K^2}{2K+1}$ of $G_j$ and $\frac{2K^2 + 2K}{2K+1}$ of dummy goods, $j \in [m]$.
	\item Agent $A_j$ takes $\frac{K}{2K+1}$ of $G_j$ and $\frac{K+1}{2K+1}$ of $G_{\phi(j), \ell(j)}$, $j \in [m]$.
	\item If $x_i = 1$, then agents in $A_{i, 1}$ obtain $\frac{5K}{4}$ of $G_{i, 2}$, $\frac{3K}{4}$ of $G_{i, 1}$; agents in $A_{i,2}$ obtain $\frac{3K}{4}$ of $G_{i,2}$, $K$ of $G_{i,3}$, $\frac{K}{4} - s_j\cdot \frac{K+1}{2K+1}$ of $G_{i, 1}$ and $s_j\cdot \frac{K+1}{2K+1}$ of the dummy good, $i \in [n]$. If $x_i=0$, then we define the allocation symmetrically, switching the groups of agents $A_{i, 1}$ and $A_{i,2}$, and the groups of goods $G_{i, 1}$ and $G_{i, 3}$.
\end{enumerate}

One can verify that this is indeed a HZ equilibrium. Agent $A_j$ has utility  $\frac{K}{2K+1} + \frac{K+1}{2K+1}\cdot \frac{5}{6} = \frac{11K + 5}{12K+6}$, and hence, the social welfare is at least $m \cdot \frac{11K + 5}{12K+6} \geq \frac{11}{12}m - O(1/m^2)$.

\paragraph{Soundness} Consider any normalized HZ equilibrium $(x,p)$. We first characterize the equilibrium behaviour of variable gadgets. In an (exact) HZ equilibrium, we say a variable gadget is {\em vacant} if no agents outside of the gadget purchase goods inside the gadget, and we call other gadgets {\em non-vacant}. Loosely speaking, only non-vacant gadgets are of interest, as vacant gadgets do not interact with the rest of market and their utility is negligible.

\begin{lemma}
	\label{lem:np-variable}
	For any $i \in [n]$, suppose the $i$-th variable gadget is non-vacant. Then the equilibrium price is one of the following three cases.
	\begin{align}
	 &p(G_{i,1}) = \frac{4}{5} \pm O\left(\frac{1}{m^2}\right), \quad p(G_{i,2}) = \frac{8}{5} \pm O\left(\frac{1}{m^2}\right), \quad p(G_{i,3}) = 0 \pm O\left(\frac{1}{m^2}\right)\label{eq:price-case1}\\
	  \text{or}  \quad	& p(G_{i,1}) = 0 \pm O\left(\frac{1}{m^2}\right),\quad  p(G_{i,2}) = \frac{8}{5} \pm O\left(\frac{1}{m^2}\right), \quad p(G_{i,3}) = \frac{4}{5} \pm O\left(\frac{1}{m^2}\right) \label{eq:price-case2}\\
	\text{or}  \quad	& p(G_{i,1}) = \frac{2}{3} \pm O\left(\frac{1}{m^2}\right), \quad p(G_{i,2}) = \frac{4}{3} \pm O\left(\frac{1}{m^2}\right), \quad p(G_{i,3}) = \frac{2}{3} \pm O\left(\frac{1}{m^2}\right) \label{eq:price-case3}
	\end{align}
\end{lemma}
\begin{proof}
	We have $p(G_{i, 2}) > 1$, otherwise $G_{i, 2}$ is oversold. Since there are outside agents that purchase goods inside the gadget, we conclude that one of the agents in $A_{i,1}$, $A_{i,2}$ must purchase goods outside the gadget, i.e. those zero-price zero-utility goods. 
Assume that some agents in $A_{i,1}$ purchase such goods (the other case is symmetric.)

	It is easy to see then that $p(G_{i, 1}) > 0$, and therefore, the optimal bundles of agents in $A_{i, 1}$ contain $G_{i,1}, G_{i,2}$ and zero-price zero-utility goods.
	Hence, $p(G_{i,2}) = 2p(G_{i,1})$. We further divide into two cases based on the optimal bundle of $A_{i,2}$.
	
	First, suppose agents in $A_{i,2}$ purchase $G_{i, 2}$, $G_{i,3}$ and zero-price zero-utility goods. Then we have $p(G_{i,2}) = 2p(G_{i,3})$. Note that there are at most $m$ agents outside the gadget that have nonzero utility for $G_{i, 1}$ or $G_{i, 3}$, and no such agents for $G_{i, 2}$. Hence, the agents in $A_{i,1}, A_{i,2}$ buy all the goods in the gadget except for at most $m$ units of $G_{i, 1}$ and $G_{i, 3}$.	 Therefore,  
	\begin{align*}
	(K \pm O(m))p(G_{i,1}) + 2K p(G_{i,2}) + (K \pm O(m)) p(G_{i,3}) = 4K. \label{eq:constraint1}
	\end{align*}
Solving the system of this and the previous two equations yields Eq.~\eqref{eq:price-case3}. 
	
	Second, suppose agents in $A_{i,2}$ only buy $G_{i, 2}$, $G_{i,3}$. Suppose $p(G_{i, 3}) > 0$. Then we know agents $A_{i,2}$ buy $K\pm O(m)$ of $G_{i,3}$, and therefore, $K\pm O(m)$ of $G_{i,2}$. Hence, we have
	\begin{align}
	 &(K \pm O(m)  ) p(G_{i,2}) + (K \pm O(m)  )p(G_{i,3}) = 2K \notag\\
\text{and} \quad  &(K \pm O(m)  )p(G_{i,1}) + 2K p(G_{i,2}) +  (K \pm O(m)  )p(G_{i,3}) = 4K. 
	\end{align}
	These equations, together with $p(G_{i,2}) = 2p(G_{i,1})$ yield the same solution, i.e., Eq.~\eqref{eq:price-case3}.
	
	Finally, assume  $p(G_{i, 3}) = 0$. Then we only have the constraint Eq.~\eqref{eq:constraint1}. These equations, together with
$p(G_{i,2}) = 2p(G_{i,1})$ yield Eq.~\eqref{eq:price-case1}.
	The symmetric case, where some agents of $A_{i,2}$ buy some goods outside the gadget, yields  Eq.~\eqref{eq:price-case2}.
\end{proof}

The following lemma follows a similar argument of Lemma~\ref{lem:goode}, we omit the proof
\begin{lemma}
	\label{lem:np-price2}
	The price of goods $G_{j}$ satisfies $p(G_{j}) = 2 + O({1}/{m^2})$, $j \in [m]$. 
\end{lemma}

We are now ready to wrap up the proof of soundness. 
Given an equilibrium $(x, p)$ that (approximately) maximizes the social welfare, we look at each non-vacant variable gadget. 
Based on the three cases stated in Lemma~\ref{lem:np-variable}, we extract the $i$-th variable to be $1$ if Eq.~\eqref{eq:price-case2} holds and $0$ if Eq.~\eqref{eq:price-case1} holds. We do nothing for the case of Eq.~\eqref{eq:price-case3} and those vacant variables (gadgets).

The total utility of all agents in $A_{i, \ell}$ , $i \in [n], \ell \in [2]$, and all agents in $A_{j, *}$,  $j \in [m]$  is at most $O({1}/{m^2})$. We focus on the utility of agents $A_j$, $j \in [m]$.
 If the $j$-th clause is satisfied, then one of the $5/6$ utility goods has zero price, and one can see that the utility is (at most) 
 \[
 \left(\frac{1}{2} \pm O\left(\frac{1}{m^2}\right)\right)\cdot 1 + \left(\frac{1}{2} \pm O\left(\frac{1}{m^2}\right)\right)\cdot \frac{5}{6} = \frac{11}{12} \pm O\left(\frac{1}{m^2}\right).
 \] 
On the other hand, if the $j$-th clause is not satisfied, we still don't need to consider the vacant gadgets (as there is no interactions), and the $5/6$ utility goods have price at least $({2}/{3})\pm O( {1}/{m^2})$.
Hence the utility is at most 
\[
\left(\frac{1}{4} \pm O\left(\frac{1}{m^2}\right)\right)\cdot 1 + \left(\frac{3}{4} \pm O\left(\frac{1}{m^2}\right)\right)\cdot \frac{5}{6} = \frac{7}{8} \pm O\left(\frac{1}{m^2}\right).
\]
Thus, if the truth assignment satisfies at most $(\frac{7}{8}+\eps) m$ clauses then the social welfare is at most 
\[
\left(\frac{7}{8}+\eps\right)m \cdot \left(\frac{11}{12} + O\left(\frac{1}{m^2}\right)\right) + \left(\frac{1}{8}-\eps\right)m\cdot \left(\frac{7}{8} + O\left(\frac{1}{m^2}\right)\right) + O\left(\frac{1}{m}\right) = \frac{175}{192}m +\frac{1}{24}\eps m + O\left(\frac{1}{m}\right).
\] 

From \cite{haastad2001some}, it is NP-hard to distinguish the case that all clauses can be satisfied (in which case there is an equilibrium with social welfare $\frac{11}{12}m -O(1/m^2)$) from the
case that at most $(\frac{7}{8}+\eps) m$ clauses can be satisfied
(in which case the maximum social welfare is at most $\frac{175}{192}m +\frac{1}{24}\eps m + O(1/m)$). The theorem follows.

The construction can be easily modified, if desired, so that all
utilities are in $[0,1]$, and every agent has minimum utility 0 and maximum utility 1.


\section{Discussion}

In this paper we resolved the complexity of computing an approximate equilibrium in the Hylland-Zeckhauser scheme for one-sided matching markets: we showed that the problem is PPAD-complete, and this holds even for inverse polynomial approximation and four-valued utilities. We leave open the complexity of exact equilibria, in particular whether the problem is FIXP-complete. Another open question is whether the PPAD-hardness of the approximation problem holds also for 3-valued utilities.

\bibliographystyle{alpha}
\bibliography{ref}

\begin{thebibliography}{GMVY17}

\bibitem[AJKT17]{akt17}
Saeed Alaei, Pooya Jalaly~Khalilabadi, and Eva Tardos.
\newblock Computing equilibrium in matching markets.
\newblock In {\em Proceedings of the 2017 ACM Conference on Economics and
  Computation}, pages 245--261, 2017.

\bibitem[AS98]{as98}
Atila Abdulkadiroglou and Tayfun Sonmez.
\newblock Random serial dictatorship and the core from random endowments in
  house allocation problems.
\newblock {\em Econometrica}, 66(3):689--702, 1998.

\bibitem[BM01]{bm01}
Anna Bogomolnaia and Herv{\'e} Moulin.
\newblock A new solution to the random assignment problem.
\newblock {\em Journal of Economic theory}, 100(2):295--328, 2001.

\bibitem[Bud11]{budish11}
Eric Budish.
\newblock The combinatorial assignment problem: Approximate competitive
  equilibrium from equal incomes.
\newblock {\em Journal of Political Economy}, 119(6):1061--1103, 2011.

\bibitem[CDDT09]{cddt09}
Xi~Chen, Decheng Dai, Ye~Du, and Shang-Hua Teng.
\newblock Settling the complexity of {A}rrow-{D}ebreu equilibria in markets
  with additively separable utilities.
\newblock In {\em 2009 50th Annual IEEE Symposium on Foundations of Computer
  Science}, pages 273--282. IEEE, 2009.

\bibitem[CPY17]{cpy17}
Xi~Chen, Dimitris Paparas, and Mihalis Yannakakis.
\newblock The complexity of non-monotone markets.
\newblock {\em J. {ACM}}, 64(3):20:1--20:56, 2017.

\bibitem[CT09]{ct09}
Xi~Chen and Shang-Hua Teng.
\newblock Spending is not easier than trading: on the computational equivalence
  of {Fisher} and {Arrow-Debreu} equilibria.
\newblock In {\em International Symposium on Algorithms and Computation}, pages
  647--656. Springer, 2009.

\bibitem[DK08]{dk08}
Nikhil~R Devanur and Ravi Kannan.
\newblock Market equilibria in polynomial time for fixed number of goods or
  agents.
\newblock In {\em 2008 49th Annual IEEE Symposium on Foundations of Computer
  Science}, pages 45--53. IEEE, 2008.

\bibitem[DPSV08]{dpsv08}
Nikhil~R Devanur, Christos~H Papadimitriou, Amin Saberi, and Vijay~V Vazirani.
\newblock Market equilibrium via a primal--dual algorithm for a convex program.
\newblock {\em Journal of the ACM (JACM)}, 55(5):22, 2008.

\bibitem[EMZ19a]{emz19a}
Federico Echenique, Antonio Miralles, and Jun Zhang.
\newblock Constrained pseudo-market equilibrium.
\newblock {\em arXiv preprint arXiv:1909.05986}, 2019.

\bibitem[EMZ19b]{emz19b}
Federico Echenique, Antonio Miralles, and Jun Zhang.
\newblock Fairness and efficiency for probabilistic allocations with
  endowments.
\newblock {\em arXiv preprint arXiv:1908.04336}, 2019.

\bibitem[EY10]{ey10}
Kousha Etessami and Mihalis Yannakakis.
\newblock On the complexity of {N}ash equilibria and other fixed points.
\newblock {\em SIAM Journal on Computing}, 39(6):2531--2597, 2010.

\bibitem[GMVY17]{gmv17}
Jugal Garg, Ruta Mehta, Vijay~V Vazirani, and Sadra Yazdanbod.
\newblock Settling the complexity of {L}eontief and {PLC} exchange markets
  under exact and approximate equilibria.
\newblock In {\em Proceedings of the 49th Annual ACM SIGACT Symposium on Theory
  of Computing}, pages 890--901, 2017.

\bibitem[GTV20]{gtv20}
Jugal Garg, Thorben Tr{\"o}bst, and Vijay~V Vazirani.
\newblock An arrow-debreu extension of the {H}ylland-{Z}eckhauser scheme:
  Equilibrium existence and algorithms.
\newblock {\em arXiv preprint arXiv:2009.10320}, 2020.

\bibitem[H{\aa}s01]{haastad2001some}
Johan H{\aa}stad.
\newblock Some optimal inapproximability results.
\newblock {\em Journal of the ACM (JACM)}, 48(4):798--859, 2001.

\bibitem[HMPY18]{he18}
Yinghua He, Antonio Miralles, Marek Pycia, and Jianye Yan.
\newblock A pseudo-market approach to allocation with priorities.
\newblock {\em American Economic Journal: Microeconomics}, 10(3):272--314,
  2018.

\bibitem[HZ79]{hz79}
Aanund Hylland and Richard Zeckhauser.
\newblock The efficient allocation of individuals to positions.
\newblock {\em Journal of Political economy}, 87(2):293--314, 1979.

\bibitem[Jai07]{jain07}
Kamal Jain.
\newblock A polynomial time algorithm for computing an {Arrow-Debreu} market
  equilibrium for linear utilities.
\newblock {\em {SIAM} J. Comput.}, 37(1):303--318, 2007.

\bibitem[Le17]{le17}
Phuong Le.
\newblock Competitive equilibrium in the random assignment problem.
\newblock {\em International Journal of Economic Theory}, 13(4):369--385, 2017.

\bibitem[McL18]{mclennan18}
Andy McLennan.
\newblock Efficient disposal equilibria of pseudomarkets.
\newblock In {\em Workshop on Game Theory}, volume~4, page~8, 2018.

\bibitem[Mou18]{moulin18}
Herv{\'e} Moulin.
\newblock Fair division in the age of internet.
\newblock {\em Annual Review of Economics}, 2018.

\bibitem[Orl10]{orlin10}
James Orlin.
\newblock Improved algorithms for computing {F}isher's market clearing prices.
\newblock In {\em Proc. 42nd ACM Symp. Theory of Computing}, pages 291--300,
  2010.

\bibitem[PP21]{pp21}
Christos Papadimitriou and Binghui Peng.
\newblock Public goods games in directed networks.
\newblock In {\em Proceedings of the 22nd ACM Conference on Electronic
  Commerce}, 2021.

\bibitem[VY11]{vy11}
V.~V. Vazirani and M.~Yannakakis.
\newblock Market equilibria under separable, piecewise-linear, concave
  utilities.
\newblock {\em Journal of the ACM}, 58(3), 2011.

\bibitem[VY21]{vy21}
Vijay~V Vazirani and Mihalis Yannakakis.
\newblock Computational complexity of the {H}ylland-{Z}eckhauser scheme for
  one-sided matching markets.
\newblock In {\em 12th Innovations in Theoretical Computer Science Conference
  (ITCS 2021)}. Schloss Dagstuhl-Leibniz-Zentrum f{\"u}r Informatik, 2021.

\end{thebibliography}

\newpage
\appendix
\ifdefined\isArxiv
\else
\graphicspath{{./Figures/}}
\fi

\section{Padding}\label{sec:padding}

Recall that the goal is to reduce the problem of finding a
  $(1/n^5)$-approximate HZ equilibrium to that of finding a $(1/n^c)$-approximate
  HZ equilibrium, for some constant $c>0$.
Let $M$ be an HZ market with $n$ goods, $n$ agents and utilities $u_{i,j}$.
Let $N=n^{5/c}$. We create a new market $M^*$ with $nN$ goods and $nN$ agents.
This is done by replacing each good $j$ in $M$ by a group
  of $N$ goods $G_j$ and replacing each agent $i$ in $M$ by a group
  of $N$ agents.
Agents in each group $G_i$ have the same utility $u_{i,j}$ for goods in $G_j$.
Given that $c$ is a constant, $M^*$ can be constructed in polynomial time.

Now let $(x^*,p^*)$ be an $\eps$-approximate HZ equilibrium of $M^*$ with
$$
\eps=\frac{1}{(nN)^{c}}\le \frac{1}{n^5}.
$$
We derive a pair $(x,p)$ for the original market $M$ as follows:
  set $p_j$ to be the minimum price of goods in $G_j$; set $x_{i,j}$ to be 
  the total allocation of goods in $G_j$ to agents in $A_i$ divided by $N$.
We prove that $(x,p)$ is an $\eps$-approximate HZ equilibrium of $M$.
The first two conditions hold trivially by the uniform scaling.
For the third condition we note that the total cost of bundles of
  $A_i$ in $(x^*,p^*)$ is at most $N(1+\eps)$.
Therefore the cost of the bundle $x_i$ of agent $i$ is at most $1+\eps$
  because prices of goods can only go down.
For the last property, we note that the utility of agent $i$ from $x_i$
  is the same as the total utility of agents in $G_i$ divided by $N$.
On the other hand, the LP for agent $i$ (with respect to $p$)
is the same as the LP for each agent 
  in $G_i$ (with respect to $p^*$), after removing subsumed constraints.
So they have the same optimal value. This finishes the correctness proof of the reduction.
\section{Disconnected Equilibria}
\label{sec:disconnectedEq}

\begin{figure}[!ht]
    \centering
    \includegraphics[width=.6\textwidth]{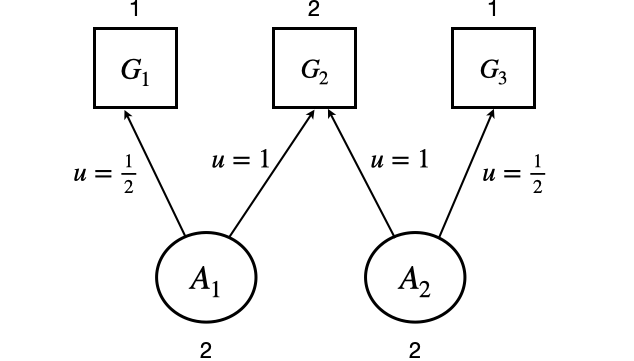}
    \vspace{0.2cm}\caption{Market with disconnected equilibria.}
    \label{fig:galaxy1}
\end{figure}

We provide a simple example showing that there exists disconnected equilibria even when the HZ market contains only four agents (i.e., $n = 4$) and the utility is drawn from $\{0, 1/2, 1\}$. The example is indeed the variable gadget we used in Section~\ref{sec:np}, we present it here for completeness.

Consider the following HZ market. There are two groups of agents, $A_1$ and $A_2$, and there are three groups of goods, $G_1$, $G_2$ and $G_3$. Let $|A_1| = |A_2| = 2$, $|G_1| = |G_3| =1$ and $|G_2| = 2$. Each agent in $A_1$ has utility $1/2$ for goods $G_1$, utility $1$ for goods $G_2$ and utility $0$ for goods $G_3$. Each agent in $A_2$ has utility $0$ for goods $G_1$, utility $1$ for goods $G_2$ and utility $1/2$ for goods $G_3$.

In any (exact) equilibrium, all goods in the same group
must clearly have the same price, because otherwise the most expensive
good in the group will remain unsold.  
Let $p = (p_1, p_2, p_3)$ denote an equilibrium price vector of $G_1, G_2, G_3$, with $\min_i p_i =0$.
Since both $A_1$ and $A_2$ have utility $1$ for $G_2$, the price $p_2 > 1$, because otherwise $G_2$ will be oversold.
\begin{claim}
	There are three disconnected equilibria in the above HZ market, with equilibria prices $(0,2,0)$, $(0, 8/5, 4/5)$ and $(4/5, 8/5,0)$ respectively. 
\end{claim}

\begin{proof}
	Since goods $G_1$ and $G_3$ are symmetric, w.l.o.g., we can assume $p_1 = 0$. The optimal bundle of $A_1$ contains exactly goods $G_1$ and $G_2$ in this case, and we know agents $A_1$ purchase $2/p_2$ unit of goods $G_2$ and $2 - 2/p_2$ unit of goods $G_1$ in total.
	We conclude $p_2 \leq 2$ since there is at most $1$ unit of goods $G_1$. 
	
	When $p_2 = 2$, agents $A_1$ get $1$ unit of $G_1$ and $1$ unit of $G_2$, and therefore, agents $A_2$ get $1$ unit of $G_2$ and $1$ unit of $G_3$. We have $p_3 = 0$ in this case. Thus, the price vector in
	this case is $(0,2,0)$.
	
	On the other hand, when $p_2 < 2$, the optimal bundle of agent $A_2$ must contain all three goods. Hence, we have $p_2 = 2p_3$ and $2p_2 +  p_3 = 4$. This leads to $p_2 = 8/5$ and $p_3 = 4/5$. The equilibrium allocation of $A_1$ equals $(3/4,5/4, 0)$ for agents $A_1$ and $(1/4, 3/4, 1)$ for agents $A_2$. 
The equilibrium price vector in this case is $(0, 8/5, 4/5)$.
Symmetrically, when $p_3=0$, there is an equilibrium price vector
	$(4/5, 8/5,0)$.	
\end{proof}

\begin{remark} If we consider also unnormalized prices (i.e. include price vectors with $\min_i p_i >0$), the set of equilibrium price vectors consists of three disjoint regions: 
$\{ (p, 2-p, p) | p \geq 0 \}$,  $\{ (1-q,1+3q,1-5q) |  0 < q \leq 1/5 \}$,  and $\{ (1-5q,1+3q,1-q) |  0 < q \leq 1/5 \}$.
The equilibrium allocations are the same as in the normalized price vectors in the three cases.
\end{remark}

\end{document}